\documentclass[sigconf]{acmart}

\renewcommand\footnotetextcopyrightpermission[1]{} 
\usepackage{amsmath}
\usepackage{amsfonts}
\usepackage{multirow}
\usepackage{graphicx}
\usepackage{multicol}
\usepackage{rotating}
\usepackage{subcaption}
\usepackage{verbatim}
\usepackage{rotating}
\usepackage{collcell}
\usepackage{color, colortbl}
\usepackage{dsfont}

\usepackage{xcolor}
\usepackage{mathtools}
\usepackage{amsthm}

\newcommand{\gs}{\mathit{G}}
\newcommand{\sgs}{\mathit{SG}}
\newcommand{\ts}{\mathit{T}}

\newcommand{\ind}[1]{\mathds{1}[#1]}
\newcommand{\fmin}{F_{\scriptscriptstyle\!\mathrm{min}}}

\newcommand{\supp}[1]{\Omega\left(#1\right)}

\newtheorem{proposition}{Proposition}
\newtheorem{lemma}{Lemma}
\definecolor{lgreen}{rgb}{0.84, 1, 0.88}

\usepackage{url}

\begin{document}
%
\title{Password Guessers Under a Microscope: An In-Depth Analysis to Inform Deployments}

\author{Zach Parish}
\affiliation{
 \institution{Ontario Tech University}
 \country{Ontario, Canada}
 }
\email{zachary.parish@ontariotechu.net}

\author{Connor Cushing}
\affiliation{
 \institution{Ontario Tech University}
 \country{Ontario, Canada}
 }
\email{connor.cushing@ontariotechu.net}

\author{Shourya Aggarwa}
\affiliation{
 \institution{Indian Institute of Technology Delhi}
 \country{New Delhi, Delhi, India}
 }
\email{shourya.aggarwal.cs117@cse.iitd.ac.in}

\author{Amirali Salehi-Abari}
\affiliation{
 \institution{Ontario Tech University}
 \country{Ontario, Canada}
}
\email{abari@ontariotechu.ca}

\author{Julie Thorpe}
\affiliation{
 \institution{Ontario Tech University}
 \country{Ontario, Canada}
}
\email{julie.thrope@ontariotechu.ca}

\begin{abstract}
Password guessers are instrumental for assessing the strength of passwords. Despite their diversity and abundance, comparisons between password guessers are limited to simple success rates. Thus, little is known on how password guessers can best be combined with or complement each other.  
 To extend analyses beyond success rates, we devise an analytical framework 
to compare the types of passwords that guessers generate. Using our framework, we show that different guessers often produce dissimilar passwords, even when trained on the same data. We leverage this result to show that combinations of computationally-cheap guessers 
are as effective in guessing passwords as computationally-intensive guessers, but more efficient. 
Our framework can be used to identify combinations of guessers that will best complement each other. 
To improve the success rate of any guesser, we also show how an effective training dataset can be identified for a given target password dataset, even when the target dataset is hashed. Our insights allow us to provide a concrete set of practical recommendations for password checking to effectively and efficiently measure password strength.

\end{abstract}




\maketitle

\pagestyle{plain}

\section{Introduction}

Passwords are presently the most common form of user authentication, providing the first layer of defense in most systems. User authentication aims to confirm a user's claimed identity, typically by something the user knows (e.g., a password), something they have (e.g., a mobile device), or something they are (e.g., a biometric).  
Despite decades of research into 
more secure authentication methods, passwords remain dominant mainly due to their ease to implement and familiarity to most users \cite{Bonneau_quest_2012}.  
Password systems, in spite of their popularity, suffer from many security issues as passwords are mistakenly given to attackers \cite{Thomas2017},  reused across accounts \cite{Das_tangled_2014}, and cracked by guessing attacks \cite{Weir_password_2009, Komanduri_passwords_2011, Veras_semantic_2014, Bonneau_science_2012, Durmuth_omen:_2015, Melicher_fast_2016, Wang2016, Pal2019}. 

Password guessing attacks are a threat to both accounts (particularly after a data breach of hashed passwords) and hard-disk encryption (where passwords are used as a key).
To protect against password guessing attacks, administrators are advised to perform password checking, either \textit{proactively} at the time of password creation or \textit{reactively} through attempting to crack their own password databases \cite{Bishop1995}. 
While there are many password guessers available, the administrator's choice of them is critical for effective password checking. However, there is uncertainty as to which guessers to use, and how to train them, for best results. 

To make an informed decision, an administrator must understand how password guessers compare to and complement each other under different conditions. 
Unfortunately, the literature lacks methods to support such decisions and analyses of guesser combinations and training. 
Our work aims to fill this gap, by creating and applying a framework to put a set of password guessers ``under a microscope", in order to support such decisions.

Our contributions are as follows:
(1) We create an analytical framework to reveal insights into password guessers' behavior and their ability to complement and substitute each other. Our framework is an asset in identifying sets of complimentary guessers (as shown in our experiments).
(2) We apply our framework to perform a comprehensive comparison between a set of six popular password guessers, across a variety of training conditions. This comparison is arguably the most comprehensive to date, as it compares many aspects of the password guessers, including how they complement each other, how well they generalize, how sensitive they are to training data size, and their success rate over six different training and testing datasets.
(3) We show how practitioners can get more bang for their buck by using combinations of computationally-cheap guessers that, when used together, have comparable success rates to computationally-intensive guessers, but are more efficient (i.e., run faster). 
(4) We perform a comprehensive analysis of six publicly leaked password datasets, to support our investigation on how guesser performance is impacted by different aspects of training data.
(5) We describe how a useful similarity metric can be applied to identify a similar (which our results support is best) training dataset for password guessers, even when the target dataset only contains  hashed passwords.

Our work has two primary outcomes: (i) Our results allow us to provide a set of recommendations for practitioners performing password checking. (ii) Our analytical framework supports more comprehensive comparisons between password guessers. We discuss use cases regarding how researchers and practitioners can use our framework to understand how additional password guessers can compliment or substitute others.

\section{Related Work}
Unfortunately, it has been repeatedly shown that user passwords are often similar or identical, and are consequently guessable by an adversary \cite{Malone2011, Castelluccia_adaptive_2012, Mazurek_measuring_2013, Bonneau_science_2012}. In this section, we review some security concerns with passwords, their counter-measures, and finally how our work fits into the literature. 

\vskip 3mm
\noindent \textit{Patterns in Passwords.}
Many users adopt common strategies for creating their passwords to help them remember their passwords. However, these strategies leave behind specific patterns, which often make passwords more guessable. These patterns include keyboard patterns \cite{Schweitzer_Visualizing_2006}, distribution of character classes (or password structures) \cite{Weir_testing_2010}, replacement of letters with resembling characters (e.g., e to 3) \cite{Jakobsson2013}, popular topics (e.g., love) \cite{Veras_semantic_2014} and dates \cite{Veras2012}. 


\vskip 2mm
\noindent \textit{Reuse of Passwords.}
Password reuse weakens password strength. When a password is reused across multiple accounts, the breach of a password in one account could lead to a breach of other accounts. The average password is used for approximately $6$ different websites \cite{Florencio2007}, and $77\%$ of users either reuse or modify an existing password \cite{Das_tangled_2014}. These reused passwords have been exploited in targeted attacks (i.e., against a single target user), with success ranging from 16\% in 1000 guesses \cite{Pal2019} to 32-73\% in 100 guesses when personal information is also incorporated \cite{Wang2016} . 


\vskip 2mm
\noindent \textit{Password Composition Policies.} To prevent users from selecting weak passwords, many systems implement password composition policies---sets of rules that a new acceptable password must follow. Common examples of composition policies include a minimum password length and/or the inclusion of characters from multiple character classes (e.g., lowercase, uppercase, numbers, special characters). Despite their practical benefit in strengthening selected passwords \cite{Summers2004,Komanduri_passwords_2011}, overly strict password policies push users to insecure behaviors \cite{Campbell2011,Inglesant2010,Komanduri_passwords_2011} including writing down passwords \cite{Inglesant2010}, reusing passwords \cite{Komanduri_passwords_2011,Campbell2011}, or extending a weak password with a special character \cite{Komanduri_passwords_2011}.  Partly due to this usability shortfall, many social-media websites, which are often targets of attacks, choose to adopt less restrictive policies \cite{Florencio2010}.

\vskip 2mm
\noindent \textit{Password Meters.} Password meters, by estimating the strength of passwords during creation, encourage users to create stronger passwords 
\cite{Ur_password_2012}. However, most of the heuristic-based meters used in practice don't accurately reflect actual password strength \cite{DeCarnedeCarnavalet2014}. Recent developments focus on various approaches, such as advanced heuristics-based methods \cite{wheeler2016zxcvbn}, probabilistic methods (e.g., Markov model) \cite{Castelluccia_adaptive_2012}, and neural networks \cite{Melicher_fast_2016,Ur_design_2017,Pal2019}. Proposals based on neural networks, Markov models, and PCFGs have been found to outperform others \cite{Golla_2019_Meters}. Also, password meters can be personalized either by taking into account a user's personal information (e.g., user profile \cite{Ji_password_2017} or previously-leaked passwords \cite{Pal2019}) in measuring the password strength, or by providing personalized feedback for password strength improvement \cite{Ur_design_2017}.


\vskip 2mm
\noindent \textit{Password Guessing Tools.} 
There are many widely-studied guessing tools and techniques for guessing passwords.
\emph{Markov models} have been promising in password guessing \cite{Narayanan_fast_2005, Durmuth_omen:_2015}. \emph{Probabilistic context-free grammars (PCFGs)} \cite{Weir_password_2009} (and its extensions \cite{Houshmand_2015_Next,PCFG-code})  create grammar structure-based password guesses, and has been widely-used (see, for example \cite{Bonneau_science_2012, Kelley_guess_2012, Das_tangled_2014, Ur_password_2012, Mazurek_measuring_2013, Castelluccia_adaptive_2012}). The \emph{semantic guesser} \cite{Veras_semantic_2014} expanded PCFGs to exploit semantic patterns in passwords. Recently, \emph{neural network guessers} have drawn considerable attention \cite{Melicher_fast_2016,Hitaj_2019_PassGAN}. The use of multiple guessers has been proposed to measure password strength \cite{Ur_measuring_2015}.  While some guessers employ a combination (e.g., PCFGv4 uses OMEN), it is not clear how to confirm they are using the most complementary guessers, nor are there any studies or methods to support their identification. Many password guessers need to be carefully tuned on training datasets to effectively guess passwords of a target dataset. Some password guessers are sensitive to language differences in training data \cite{Ji-TDSC-2017}, and the similarity between training and target datasets improves guessing success \cite{Ji_password_2017}, a finding that we corroborated in just one of our many experiments but using a different method, more data, and more guessers (see Section \ref{subsec:datasetSim}).However, it was not clear how to identify similar data sets when the target is hashed; we describe a method to do so using our methods in Section \ref{sec:discussion:usecases}.

We note that practitioners need to make many decisions to implement effective password checking. These decisions include which subset of guessing tools to choose among many available options, and which training dataset to choose. To support these decisions, the literature falls short in systematically understanding guesser behaviors and their ability to complement or substitute one another. This work attempts to address this gap. 

\section{Analytical Framework} \label{sec:framework}

The analytical framework presented in this section can be applied to evaluate any set of guessers.  It can also be used to evaluate a set of training datasets to identify the best training datasets for password checking.  These two use cases are discussed further in Section \ref{sec:discussion}.

We consider a set of $m$ password guessers $\mathcal{G} = \{g_1, \dots,g_m\}$ where each $g_i$ represents a specific guesser (e.g., John the Ripper, OMEN, etc.). We aim to understand how each guesser $g_i \in \mathcal{G}$ behaves when trained on or tested against particular password datasets, what types of passwords they guess, and how similar one guesser's behavior is to others. To this end, each guesser $g \in \mathcal{G}$ will be trained on and tested against a set of $n$ password datasets $\mathcal{D} = \{D_1, ..., D_n\}$, where each $D_j$ is a publicly-available password dataset (e.g., RockYou, Twitter, etc.).\footnote{We use the terminology of ``testing against a dataset'' when a guesser is guessing the passwords of a target password dataset.} When a guesser $g_i \in \mathcal{G}$ is trained on a dataset $D_j \in \mathcal{D}$, it can create a password guess list $L_{ij}$. To compare various guessers trained on various datasets, we develop some statistics (see Section \ref{subsec:statistics}) for comparing guessers' guess lists. Our statistics deploy some pairwise-comparison metrics (see Section \ref{subsec:metrics}), which use either structural features (see Section \ref{subsec:features}) or the passwords shared between two lists.





\subsection{Password Features} \label{subsec:features}

For each password $w$, we extract two structural features: \emph{password length} $n_w$ (i.e., the number of its characters) and the \emph{number of character classes} $c_w$ that it contains. 
We focus on four distinct character classes: lowercase letters, uppercase letters, numbers, and symbols. For instance, $w=passw0rd!$ has $n_w = 9$ and $c_w=3$ with three character classes: lowercase letter, number, and symbol.  

To extract features from password list $L'$  (e.g., leaked password database or guess list of a guesser), we first aggregate the extracted features of all $w \in L'$ into a matrix $\mathbf{V} = [v_{xy}]$ where $v_{xy}$ is the fraction of passwords in password list $L'$ which contains $y$ characters covering $x$ character classes:
\begin{equation}
    v_{xy} = \frac{1}{|L'|}\sum_{w \in L'} \ind{c_w=x \And n_w=y},
\end{equation}
where $\ind{.}$ is the indicator function, and $|L'|$ represents the number of passwords in the list.\footnote{Indicator function $\ind{s}$ returns 1 if the statement $s$ is true; otherwise 0.} The matrix $\mathbf{V}$ has a natural probability interpretation: when one selects a password $w$ from the password list $L'$ uniformly at random, the password $w$ contains $y$ characters from $x$ character classes with a probability of $v_{xy}$. In other words, our matrix $\mathbf{V}$ captures the joint probability distribution of passwords over character classes and the number of characters. To ease our notations and analyses, we collapse (i.e., flatten) the matrix $\mathbf{V}$ into a feature vector $\mathbf{v}$. We refer to this feature vector as the \emph{structural features} of a password list. This simple representation allows us to preserve the impact of password policies of each password list.

\subsection{Pairwise Comparison Metrics} \label{subsec:metrics}


Our deployed pairwise comparison metrics are symmetric, so are computed once for each pair of password lists. While these metrics can use any features, we use either the structural features described in Section \ref{subsec:features} or the passwords shared between two lists. The use of cosine similarity combined with our proposed structural password features ease the interpretation of our analyses with regard to their connections to the password policies of password datasets. The use of Jaccard Index on the passwords shared between two lists allows us to analyze the degree to which password guessers have complementary behavior on the finest level of granularity. Our metrics have been widely used in information retrieval \cite{Singhal2001, Frakes:1992:IRD:129687, Baeza-Yates:1999:MIR:553876}, data mining \cite{Berkhin2002, Dunham:2002:DMI:560701}, and other password research \cite{Ji_password_2017}. Of particular interest is the \emph{generalized Jaccard index}, which as we show in Section \ref{sec:discussion} and the Appendix, can be used to not only measure similarity between each guesser's output, but also measure similarity between a plaintext dataset and a hashed \& salted password dataset.  

\subsubsection{Cosine Similarity}\label{cosineSimilarity}
Cosine similarity measures the angle between two non-zero vectors. For comparison of two password lists, one can extract structural features from each list, and then compute the cosine similarity on the corresponding feature vectors. The cosine similarity between two password lists $A$ and $B$ is given by
\begin{equation}
    C(A, B) = \frac{\mathbf{v}_{\scriptscriptstyle\!\mathtt{A}} \cdot \mathbf{v}_{\scriptscriptstyle\!\mathtt{B}}}{\|\mathbf{v}_{\scriptscriptstyle\!\mathtt{A}}\| \|\mathbf{v}_{\scriptscriptstyle\!\mathtt{B}}\|}, \label{eq:cosine_sim}    
\end{equation}
where $\mathbf{v}_{\scriptscriptstyle\!\mathtt{A}}$ and $\mathbf{v}_{\scriptscriptstyle\!\mathtt{B}}$ are structural feature vectors of $A$ and $B$, respectively. $\|.\|$ is the Euclidean norm, and $\mathbf{v}_{\scriptscriptstyle\!\mathtt{A}} \cdot \mathbf{v}_{\scriptscriptstyle\!\mathtt{B}}$ is the dot product of those two vectors.  The closer the cosine similarity value is to 1, the smaller the angle between the two vectors is, and the more similar they are. In other words, two lists of passwords with similar feature distributions have a high cosine similarity. We use cosine similarity combined with our proposed structural features for two purposes: (i) comparing the structure of leaked password databases with each other; (ii) comparing the structure of two guess lists. 

\subsubsection{Jaccard Index}
Jaccard index measures the extent two sets overlap with each other, where the intersection of two sets is compared to their union. The Jaccard index between two password lists $A$ and $B$ can be computed by
\begin{equation}
    J(A, B)=\frac{|A \cap B|}{|A \cup B|}.
    \label{eq:jaccard_sim}
\end{equation}
The closer the Jaccard index is to 1, the closer in size the intersection of the sets is to their union, and consequently the more similar two sets are. In other words, two sets of passwords with high overlap will have a high Jaccard index. The Jaccard index also has a natural probabilistic interpretation: if one chooses a password uniformly at random from either password list, the Jaccard index captures the likelihood of selecting a password belonging to both sets. 

When password lists have duplicates (e.g., leaked password datasets), we view the password list as a multiset, a modification of sets that allows for duplicated elements. In these cases,
we apply a generalized version of the Jaccard index \cite{Ji_password_2017} to preserve the frequency information of password duplicates in password lists. 
Letting $o(w, A)$ be the number of occurrences of password $w$ in password list $A$, the \emph{generalized Jaccard index} between two password lists $A$ and $B$ is given by
\begin{equation}
    J(A, B) = \frac{\sum_{w \in U} min\left(o(w,A), o(w, B)\right)}
{\sum_{w \in U} max\left(o(w,A), o(w, B)\right)},
\label{eq:gen_jaccard_sim}
\end{equation}
where $U=(\supp{A}\cup \supp{B}$), and $\supp{A}$ represents the set of unique passwords in the password list $A$.  

\subsection{Statistics} \label{subsec:statistics}
Our comparison metrics can be readily used for the comparison of a pair of password lists. However, to compare two guessers thoroughly, it is useful to summarize the comparison metrics of two guessers under different training and testing datasets. This section explains our proposed statistics for summarizing comparison metrics. Our statistics fall into two categories: relating to guessing behaviors and relating to guessing success.



\subsubsection{Statistics for Guessing Behaviors}\label{subsec:guessingStats}
This class of statistics is devised to either compare the guessing behaviors of password guessers with each other, or measure how different training datasets affect the guessing behavior of a given guesser.

Our \emph{guessing similarity} statistic summarizes the similarity of two guessers' guess lists when trained on the same dataset by averaging the comparison metric (e.g, Jaccard or Cosine) of their guess lists over various training datasets. We calculate the \emph{guessing similarity} of two guessers $g_i$ and $g_j$ by
\begin{equation} 
    \gs(g_i, g_j, M) = \frac{1}{n} \sum_{k=1}^n M(L_{ik}, L_{jk})
    \label{eq:guessingSimilarity}
\end{equation}
where $M \in \{C,J\}$ is either Cosine similarity (see Eq.~\ref{eq:cosine_sim}) or Jaccard index (see Eq.~\ref{eq:jaccard_sim}), and $L_{ik}$ is the list of password guesses (without any duplicates) generated by $g_i$ trained on datasets $D_k$. Here, $n$ is the number of datasets in $\mathcal{D}$. We also introduce \textit{successful guessing similarity} to measure how two guessers' successful guesses are similar:
\begin{equation} \label{eq:suc-guessing-similarity}
    \sgs(g_i, g_j, M) = \frac{1}{n(n-1)} \sum_{k=1}^n\sum_{\ell\neq k}^n M\left(L_{ik}\cap D_\ell, L_{jk}\cap D_\ell\right).
\end{equation}

One might be interested in measuring how similarly two different password datasets can train guessers. To this end, we introduce our \emph{training similarity} statistic which calculates the extent two different training password datasets result in generating similar guess lists of passwords when used for training. We define \emph{training similarity} between two datasets $D_j$ and $D_k$  by
\begin{equation}\label{eq:trainingSimilarity}
    \ts (D_j, D_k, M) = \frac{1}{m} \sum_{i=1}^m M(L_{ij}, L_{ik}),
\end{equation}
where $m$ is the number of different guessers in $\mathcal{G}$. This formula computes how similarly $D_j$ and $D_k$ can train guessers on average. By capturing the extent two various datasets are effectively similar in training guessers, one can identify training datasets which are as effective as another dataset in training guessers. This could be used to identify effective, yet small datasets, which could drastically speed up the training process.



\subsubsection{Statistics for Guessing Success}
The \emph{guessing success} statistics quantify the guessing accuracy of guessers under various settings (e.g., training and testing datasets), and also determine how training data affects guessing success for various guessers.

When each guesser $g_i$ is trained on password dataset $D_j$ and tested against password dataset $D_k$, one can compute its \emph{success rate}, as the portion of successfully guessed passwords, by

\begin{equation}
   s_{ijk} = \frac{|L_{ij} \cap D_k|}{|D_k|}. 
\end{equation}
Note that $s_{ijk} \in [0,1]$, where $s_{ijk}=1$ implies that all passwords in $D_k$ are guessed successfully by $g_i$ trained on $D_j$. To summarize the success rate for a specific guesser $g_i$, one can compute its mean success rate over all distinct training and testing datasets by 
\begin{equation}\label{eq:guesserSuccessRate}
    \left\langle s_{i::} \right\rangle = \frac{1}{n(n-1)} \sum_{j=1}^n \sum_{k\neq j}^n s_{ijk}.
\end{equation}
We similarly compute the success rate of training dataset $D_j$ by 
\begin{equation}\label{eq:trainingSuccessRate}
    \left\langle s_{:j:} \right\rangle = \frac{1}{m(n-1)} \sum_{i=1}^m \sum_{k\neq j}^n s_{ijk},
\end{equation}
and the average success rate of a fixed dataset $D_j$ and guesser $g_i$ by
\begin{equation}\label{eq:FixedDandGSuccess}
    \left\langle s_{ij:} \right\rangle = \frac{1}{n-1} \sum_{k\neq j}^n s_{ijk}.
\end{equation}

\section{Experiments}
Our experiments aim to understand the impact of training dataset choice on guessers, the performance of guessers, 
and how guessers can complement or substitute one another.

\subsection{Experimental Setup} \label{subsec:experimental-setup}

We choose a variety of different password datasets and guessers.

\subsubsection{Password Datasets}\label{datasets}
Our experiments use a variety of publicly available leaked password datasets, which have been the subject of other password research studies (for example, \cite{Wei2018, Zhou2017, Wang2016, Weir_testing_2010, Das_tangled_2014, furnell2011assessing}). We have curated and cleaned these datasets by converting their passwords to Unicode. Table \ref{table:passSetSize} shows the number of total and unique passwords in each dataset as well as the ratio between those values.\footnote{We exclusively use publicly available datasets and don't report any specific password information. Thus, there is no risk of exposing private user information. 
We keep  only the passwords with no links to their original owner.}

\begin{table}[tb]
    \centering
    \caption{The password datasets, their sizes, and the ratio between unique and total number of passwords. *Merged contains all other plaintext datasets in this table. }
    \label{table:passSetSize}
    \begin{tabular}{lllll}
    \toprule
    \multicolumn{1}{c}{} & \multicolumn{2}{c}{\textbf{Number of Passwords}} \\ 
        \cmidrule[0.7pt](lr){2-3}
        \textbf{Datasets} & \textbf{Total} & \textbf{Unique} & \textbf{Ratio} & \textbf{Type}\\
        \midrule
        ClixSense \cite{goodin_million_2016} & 2,222,359 & 1,628,205 & 0.7326 & Plaintext\\
        Webhost \cite{foxbrewster_million_2017} & 15,292,021 & 10,589,775 & 0.6925 & Plaintext\\
        Mate1 \cite{russon_mate1_2016} & 27,403,932 & 11,988,154 & 0.4375 & Plaintext\\
        RockYou \cite{cubrilovic_rockyou_2009} & 32,596,319 & 14,337,716 & 0.4399 & Plaintext\\
        Fling \cite{das_million_2016} & 40,769,652 & 16,810,091 & 0.4123 & Plaintext\\
        Twitter \cite{TwitterDataset} & 40,872,901 & 22,579,065 & 0.5524 & Plaintext\\
        Merged* & 159,157,184 & 67,628,637  & 0.4249 & Plaintext\\
        LinkedIn \cite{hackett_2016} & 174,243,105 & 61,829,207 & 0.3548 & Hashed\\
        \bottomrule
    \end{tabular}
\end{table}

\subsubsection{Password Guessers} \label{passwordGuessingTools}

To include a wide variety of guesser behaviors, we focus on six guessers from three different classes of password guessers: Markov models, Probabilistic Context Free Grammars (PCFGs), and Neural Networks. All examined guessers are used with their recommended optimal/default settings, or tuned to perform their best on our datasets.


\vskip 1.5mm
\noindent \textit{John the Ripper (JtR-Markov).} We use its community build (1.9.0-bleeding-jumbo) \cite{jtr-bleeding-jumbo} in Markov mode. 
We restrict the maximum length of passwords to 12 characters, which provided the best results and is consistent with other studies \cite{Veras_semantic_2014}. 
JtR runs single-threaded during both training 
and guessing.

\vskip 1.5mm
\noindent \textit{Ordered Markov Enumerator (OMEN).} We use OMEN \cite{Durmuth_omen:_2015, OMEN-code} with the default settings. OMEN produces only ASCII passwords and runs single-threaded during training and guessing.

\vskip 1.5mm
\noindent \textit{Probabilistic Context-Free Grammar (PCFGv4).}
We used PCFG version 4.0 \cite{PCFG-code}, an extension of the original PCFG \cite{Weir_password_2009}. This version uses OMEN to generate a certain percentage of passwords and generate the remainder with PCFGs. We have disabled this feature to generate passwords exclusively from PCFGv4 
as the use of OMEN decreased the success rate in most of our tests. 
PCFGv4 runs single-threaded.


\vskip 1.5mm
\noindent \textit{Semantic Guesser (Sem).}
We use the lite 
version of Sem \cite{Sem-code,Veras_semantic_2014}. 
The grammars are trained as recommended using maximum likelihood estimation, the backoff algorithm is used for producing tags, and mangling rules are enabled for generating guesses. 
Sem uses multiprocessing 
during training, but runs single threaded for guessing.

\vskip 1.5mm
\noindent \textit{Neural Network (NN).}
We generate guesses using the NN's ``human'' mode \cite{Melicher_fast_2016, NN-code}, and sort them in descending probability order. We limit the length of passwords to 6--40 characters to maximize NN's success rate for our datasets. We use a model consisting of three LSTM layers (with 1024 neurons each) and two dense layers (with 512 neurons each). 
The neural network is our only guesser that uses GPU resources along with CPU. The neural network runs multi-threaded during training and guessing.



\vskip 1.5mm
\noindent \textit{Identity Guesser (ID).}
This guesser takes a training dataset as input, removes its duplicates, and outputs its unique passwords in descending order of their frequency in the training dataset. 
In other words, this guesser computes the empirical probability distribution of the passwords in the training dataset (i.e., training phase), then outputs the passwords from the highest to the lowest probability (i.e., generation phase).  This simple guesser is a valuable benchmark for understanding how well other guessers learn and generalize. 


\subsection{Impact of Training Data Choice}\label{subsec:analysis_data}
We investigate how guessing success rates are impacted by different aspects of training data.
We train all six password guessers on each of the six individual plaintext datasets and test them against every other plaintext dataset, yielding 180 password cracking scenarios. For all guessers, we set the cutoff to 300 million guesses. 

Figure \ref{fig:successTraining} captures the average success rates for various pairs of training and testing datasets. One can make two important observations: (i) 
some datasets (e.g., Twitter, Mate1) are more effective training data than others (e.g., Webhost); (ii) some pairs of datasets are effective for training and testing against each other, i.e., when one dataset can train guessers well against another dataset (e.g., RockYou-Mate1,  ClixSense-Mate1, etc.). These two observations motivate us towards a  deeper analysis of the characteristics of effective training datasets.

\begin{figure}[tb]
    \centering
    \includegraphics[width=0.8\columnwidth]{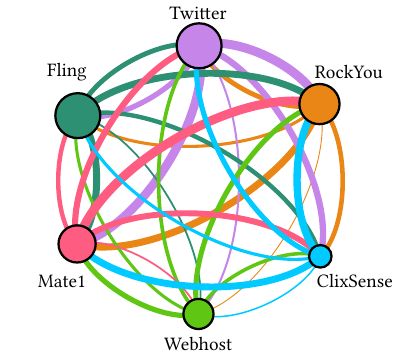}
    \caption{The average success rates over all guessers for training and testing dataset pairs. The edges are directed clockwise from training to testing dataset, with colors matching the training dataset color. The edge width is proportional to the average success rate of guessers for a fixed training and testing dataset pair.  The node size shows the dataset size.}
    \label{fig:successTraining}
    \vspace{-10pt}
\end{figure}

\begin{table}
    \begin{center}
    \caption{Mean success rates (and std.\ deviations) for training password datasets. Datasets are ordered smallest to largest.}
    \label{table:successTraining}
    \begin{tabular}[t]{lcc}
        \toprule
        \textbf{Training} & \textbf{Success Rate 1m} & \textbf{Success Rate 300m}\\
        \midrule
        ClixSense & 15.929 (12.634) & 33.737 (16.355) \\
        000webhost & 8.72 (5.387) & 29.602 (11.968) \\
        Mate1 & 18.337 (13.234) & 38.167 (14.799) \\
        RockYou & 13.845 (14.037) & 30.264 (17.592) \\
        Fling & 11.835 (9.158) & 35.155 (16.393) \\
        Twitter & 20.59 (15.303) & 42.815 (17.371) \\
        \bottomrule
        \end{tabular}
        \end{center}
\end{table} 
\subsubsection{Size of training dataset}
We ask whether the success rate of a guesser, on average, increases with the size of training dataset.  Table \ref{table:successTraining} shows the average success rates of each training dataset over all guessers and target datasets (computed by Eq.~\ref{eq:trainingSuccessRate}), with datasets ordered from smallest to largest size. While our largest dataset performs the best, our smallest dataset ClixSense outperforms both Webhost and RockYou, which are over six and fifteen times larger than it respectively. 
For a formal analysis, we calculated the statistical correlation between the number of passwords in the training dataset and the averaged success rate. The resulting Pearson coefficient of 0.189 (p= 0.315) suggests insignificant correlation between training dataset size and success rate. This result suggests that a larger dataset size isn't necessarily a requirement for an effective training data set.

\begin{figure*}[t]
\begin{center}
    \begin{subfigure}{0.47\textwidth}
    \begin{center}
    \includegraphics[width=0.8\textwidth]{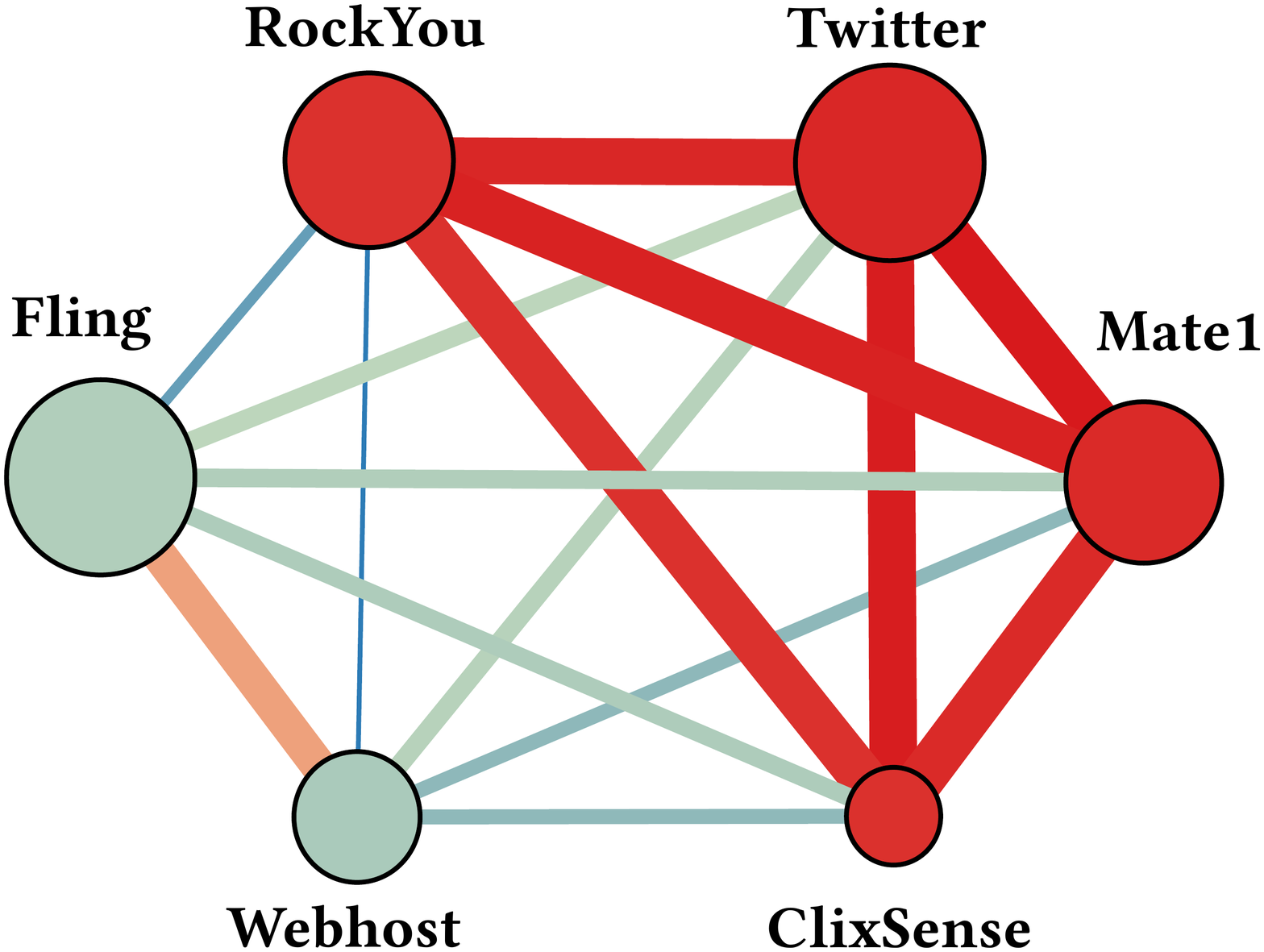}
    \includegraphics[width=0.11\textwidth]{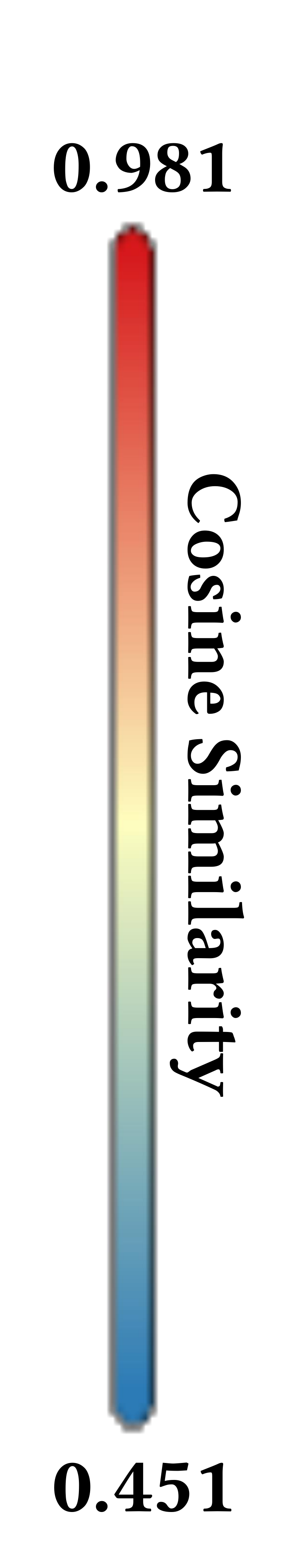}
    \caption{The cosine similarity.}
    \label{fig:passwordMetricsCosine}
    \end{center}
    \end{subfigure}
    \begin{subfigure}{0.47\textwidth}
    \begin{center}
    \includegraphics[width=0.8\textwidth]{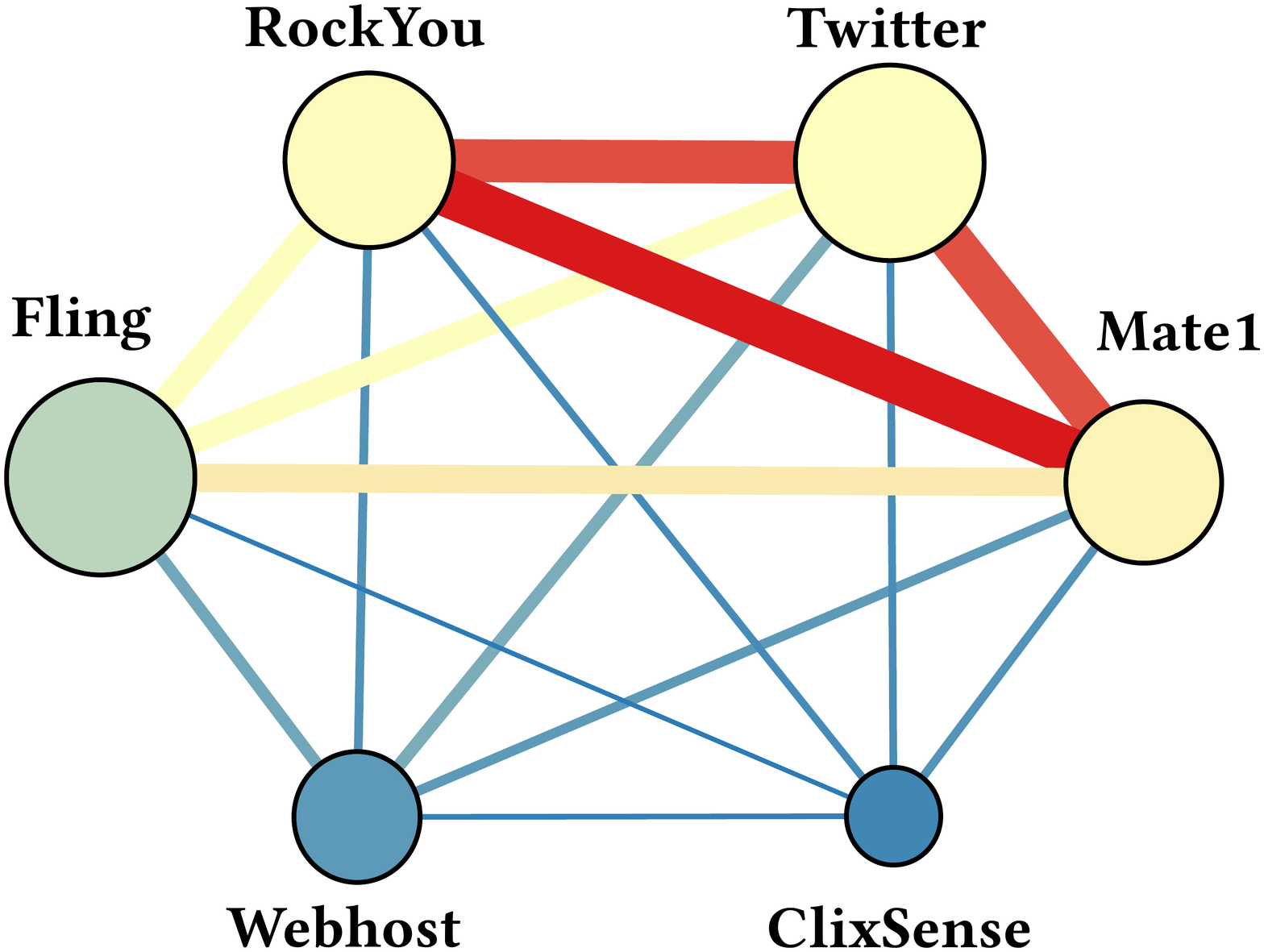}
    \includegraphics[width=0.11\textwidth]{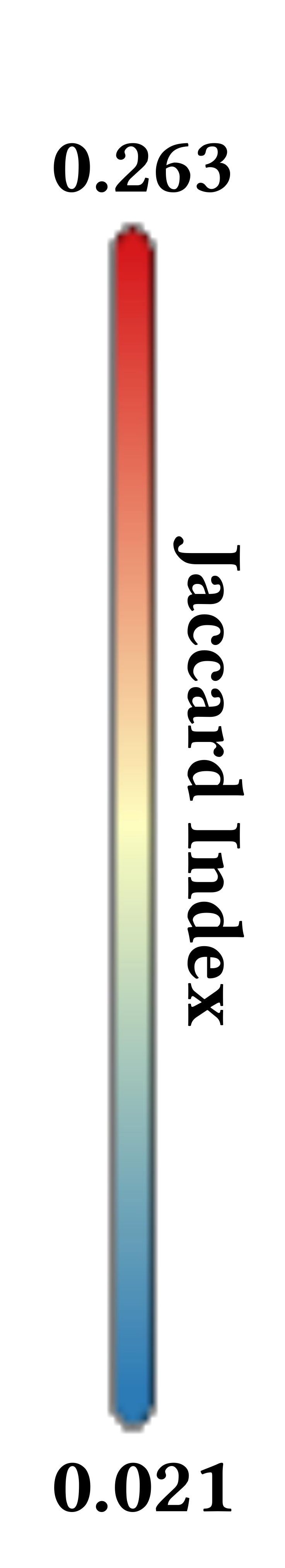}
    \caption{The generalized Jaccard index.}
    \label{fig:passwordMetricsJaccard}
    \end{center}
    \end{subfigure}\\
\begin{subfigure}{0.47\textwidth}
\begin{center}
\includegraphics[width=0.8\columnwidth]{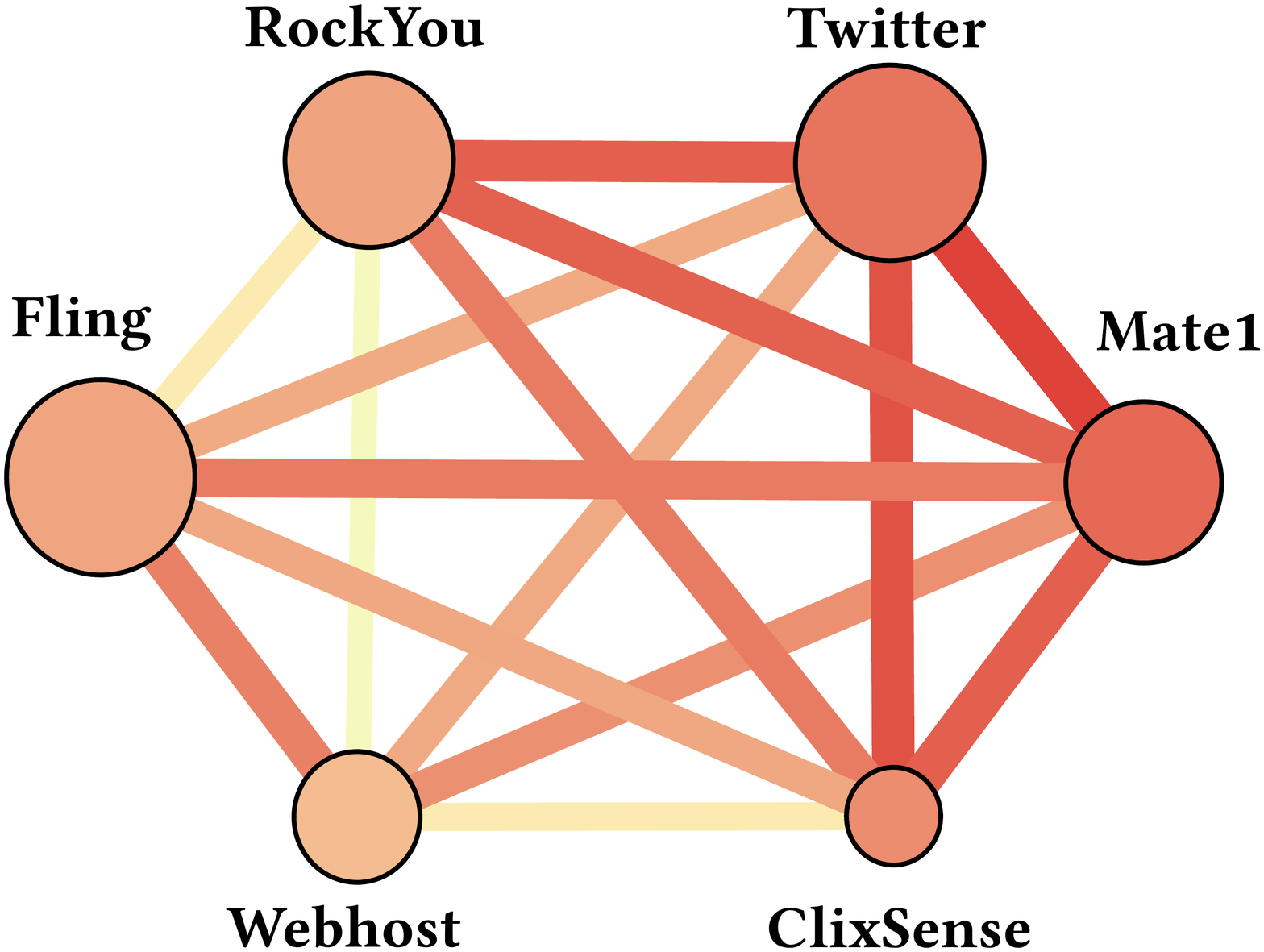}
\includegraphics[width=0.11\columnwidth]{images/ScaleBars/DatasetCosine.pdf}
\caption{The cosine training similarity.}
\label{fig:trainingSimilarityCosine}
\end{center}
\end{subfigure}
\begin{subfigure}{0.47\textwidth}
\begin{center}
\includegraphics[width=0.8\columnwidth]{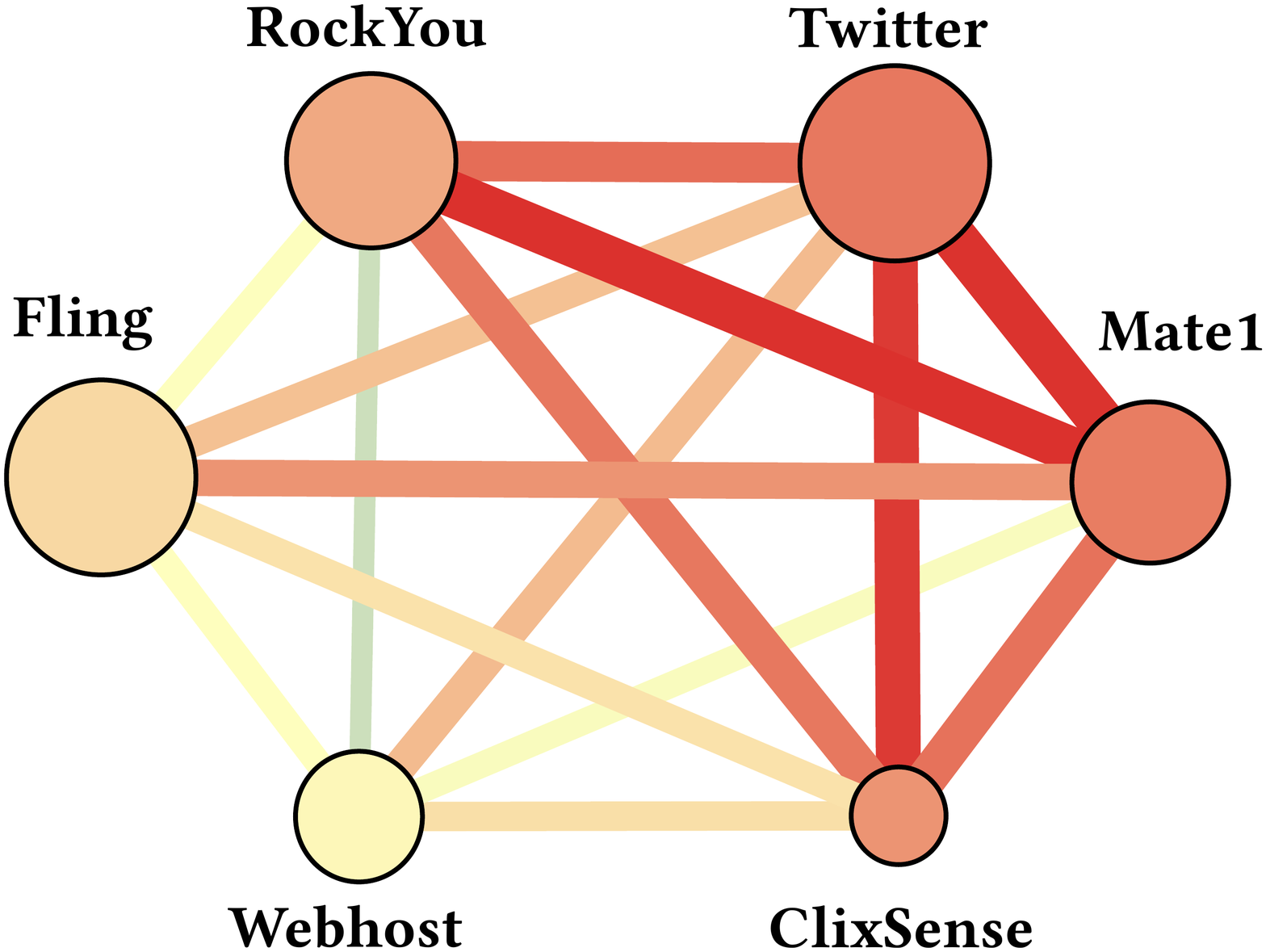}
\includegraphics[width=0.11\columnwidth]{images/ScaleBars/DatasetJaccard.pdf}
\caption{The Jaccard training similarity.}
\label{fig:trainingSimilarityJaccard}
\end{center}
\end{subfigure}
    \caption{Plaintext datasets with their pairwise (a) cosine similarity, (b) generalized Jaccard similarity, (c) cosine training similarity, and (d) Jaccard training similarity. The training similarity between datasets is computed by Eq.~\ref{eq:trainingSimilarity}. The edge weights and colors are based on the corresponding metric value between two datasets. The node color captures the metric average for the corresponding dataset. The node size is proportional to the dataset size.}
    \label{fig:passwordMetrics}
\end{center}
\end{figure*}
\vskip 2mm
\subsubsection{Similarity between training and target datasets}
\label{subsec:datasetSim}
We next focus on how the similarity between training and target datasets impacts the success rate of guessers. We first compute the cosine similarity and generalized Jaccard index (see Eq.~\ref{eq:cosine_sim} and Eq.~\ref{eq:gen_jaccard_sim}) between password datasets, and then explore the relationship of these similarities with success rates. 

Figure \ref{fig:passwordMetricsCosine} shows that that Mate1, Twitter, RockYou and ClixSense have high structural password similarity (i.e., cosine similarity). Fling and Webhost are dissimilar to other datasets, but similar to each other. 
Figure \ref{fig:passwordMetricsJaccard} suggests that the exact overlap between datasets (i.e., generalized Jaccard similarity) is often low with exceptions for larger datasets (i.e., Fling, Twitter, RockYou and Mate1), likely due to their sizes. 

The cross-examination of Figures \ref{fig:successTraining}, \ref{fig:passwordMetricsCosine}, and \ref{fig:passwordMetricsJaccard} suggest the datasets with higher similarity tend to have mutually higher success rates (e.g., Mate1 and RockYou share high similarity and mutual success rates). Thus, we hypothesize that the  similarity between training and testing datasets has a positive effect on success rate. To test this hypothesis, we ran Pearson statistical tests between the similarity metric of any ordered pair of datasets and their success rates
. Our cosine similarity and Jaccard metric have correlation coefficients of 0.597 ($p= 0.00049$) and 0.596 ($p = 0.00049$) respectively. Both are significant and large by Cohen's convention. This further confirms that dataset similarity, structural (cosine) or overlap (Jaccard), is a key factor in success rate. These results complement previous findings \cite{Ji_password_2017}  on the relationship between the similarity of training and testing datasets and guesser success rates. We note that this is our only experiment with partial overlap with other work \cite{Ji_password_2017} by computing cosine similarity and Jaccard index between datasets; however,  we use a different set of datasets and guessers, different features for cosine similarity, and a different application of Jaccard index (between datasets rather than between their features). 
We  also go on to show in Section \ref{sec:discussion:usecases} how Jaccard Index can be used to measure similarity even when the target dataset is hashed.

\subsubsection{Training similarity between datasets}
We next explore how similarly two datasets can train a guesser using our notion of training similarity (see  Eq.~\ref{eq:trainingSimilarity}). This exploration might not have a direct application in password checking, but offers interesting observations for password guessing. Our investigation is motivated by the surprising performance of ClixSense in Table \ref{table:successTraining}. Despite ClixSense's small size, its performance raises the question of how similarly ClixSense and a bigger dataset can train a guesser, as smaller training datasets may be desirable in some cases to reduce training time. We exclude the Identity guesser in this analysis due to its simplicity in learning; also, its results mirror dataset similarity (see \ref{subsec:datasetSim}).  

Figures \ref{fig:trainingSimilarityCosine} and \ref{fig:trainingSimilarityJaccard} demonstrate the cosine and Jaccard training similarity between our datasets. The cosine training similarity is relatively high between most pairs of datasets. 
The cluster of RockYou, Twitter, Mate1, and ClixSense share relatively high overlap of generated passwords (see their pairwise Jaccard training similarity). This means passwords generated from training with ClixSense, despite its small size, have high overlap with passwords generated from training with other datasets.

\subsection{Individual Guesser Performance}
To evaluate the performance of each guesser, we compute its average success rate and runtime across varied training data, target data, and password guessing scenarios (i.e., online and offline attacks). 

\subsubsection{Guessing Success Rate} To gauge the average performance of each guesser, we train and test every guesser on each possible pair of non-merged plaintext datasets. 
Then, each guesser generates guess lists at cutoffs of 1 million and 300 million guesses to simulate online \cite{FHO2016_PushingString} and limited offline attacks, respectively. Table \ref{table:guesserSuccess} shows the mean success rate of each guesser, computed by Eq.~\ref{eq:guesserSuccessRate}. At one million guesses, PCFGv4 and the Identity guesser outperform others, while JtR-Markov and OMEN perform the worst. Notably, only PCFGv4 is able to outperform the Identity guesser at this cutoff with a negligible margin. 

\begin{table}[tb]
        \begin{center}
        \caption{Guessers' mean success rates at 1 Million and 300 Million guesses (standard deviations in parenthesis). The two best and worst are highlighted with green and red, resp.}
        \label{table:guesserSuccess}
        \begin{tabular}{lcc}
        \toprule
        \textbf{Guesser}&\textbf{Success Rate@1M}&\textbf{Success Rate@300M}\\
        \midrule
        Identity &\cellcolor{lgreen}23.238 (11.859) &30.519 (14.079) \\
        JtR-Markov &\cellcolor{pink!40}0.665 (0.993) &\cellcolor{pink!40}27.591 (11.563) \\
        OMEN &\cellcolor{pink!40}5.921 (3.225) & \cellcolor{pink!40}22.121 (10.749) \\
        Sem &18.219 (10.344) &\cellcolor{lgreen}41.343 (13.274) \\
        PCFGv4 &\cellcolor{lgreen}23.551 (11.545) &\cellcolor{lgreen}47.397 (12.364) \\
        NN & 17.662 (11.585) & 40.768 (19.734) \\
        \bottomrule
        \end{tabular}
        \end{center}
\end{table}

For three-hundred million guesses, PCFGv4 performs the best, with a 6\% lead over the second best guesser Sem. The Identity guesser performs surprisingly well, with an average of 30.5\% (but a high standard deviation of 14.07\%) in at most 21,653,268 guesses (compared to 300 million guesses for other guessers).\footnote{The upperbound for number of guesses in the Identity guesser is derived from the maximum number of unique passwords in our datasets.} In its best case, the Identity guesser trained on Twitter guesses 56.7\% of RockYou, only 10.14\% lower than the best guesser PCFGv4 on that same pair. The Identity guesser's high success rate arises from a relatively large overlap between datasets, observed 
in Figure \ref{fig:passwordMetricsJaccard}. 
 OMEN under-performs JtR-Markov, performing worst overall at this cutoff.

\subsubsection{Average Runtime}
To help a system administrator understand the resource requirements of guessers, we next analyze their runtimes during training and guess list creation. Each guesser is trained and generates guesses on the same GPU-accelerated server which ran no other jobs. The server has 2 Intel(R) Xeon(R) Gold 6148 CPUs with 80 total cores @ 2.40GHz and 4 Nvidia GeForce 1080 Ti GPUs. 
We note that only the neural network benefits from multiple GPUs to parallelize computations.

\tabcolsep=0.11cm
\begin{table}[tb]
    \centering
    \caption{Guesser training and generation time. Training datasets are randomly sampled from the Merged Dataset. Guessers (except Identity) generated 300M guesses.}
    \label{table:guesserRuntime}
    \begin{tabular}{lccc}
    \toprule
    \multicolumn{1}{c}{} & \multicolumn{2}{c}{\textbf{Training}} \\ 
        \cmidrule(lr){2-3}
        \textbf{Guesser} & \textbf{1 Million} & \textbf{50 Million} & \textbf{Generation} \\
        \midrule
        JtR-Markov & 00h 00m 00.1s & 00d 00h 00m 02.2s &  00h 00m 33s\\
        Identity & 00h 00m 00.3s & 00d 00h 00m 24.9s & 00h 00m 18s \\
        OMEN & 00h 00m 03.0s & 00d 00h 00m 23.0s & 00h 07m 10s\\
        Sem & 00h 01m 38.3s & 00d 00h 20m 14.6s & 00h 55m 30s\\
        PCFGv4 & 00h 03m 49.5s & 00d 01h 03m 38.4s & 00h 30m 58s \\
        NN & 01h 18m 08.0s &  02d 17h 01m 49.0s &  19h 44m 20s \\
        \bottomrule
    \end{tabular}
\end{table}

%
Table \ref{table:guesserRuntime} reports guesser training and generation time. For training, we created two datasets by sampling 1 million and 50 million passwords from the Merged dataset.\footnote{Our code for training the identity guesser (i.e., computing empirical distribution of unique passwords) and its guess generation (i.e., sorting passwords based on their probabilities) is written in Python without any optimization.} For each guesser, the training time increases with the training dataset size. The Markov-based and Identity guessers perform the fastest ($<$ 25 sec. for 50 million), with PCFGs taking longer (about one hour for 50 million) and the neural network taking the longest (more than 2.5 days for 50 million). For password generation, we observe that the Identity guesser and Markov models are again by far the fastest. Note that the Identity guesser only produced approximately 67M guesses, almost 4.5 times fewer guesses than produced by others. The NN is considerably slower than others: 2100 times slower than JTR-Markov, and even 39 times slower than PCFGv4.

%

\subsection{Guesser Behaviors}

We investigate the behavior of each guesser (i.e., their generated guess lists) under various training and target datasets. We also explore how each guesser complements and substitutes others.

\subsubsection{Generalizability}
One important characteristic of guessers is how well they can generalize, i.e., predict and generate previously unseen passwords. To measure this, we train each guesser on the Webhost dataset as it is the least similar to the other datasets, both in terms of structure (see Figure \ref{fig:passwordMetricsCosine}) and actual password overlap (see Figure \ref{fig:passwordMetricsJaccard}). We then test the Webhost trained guessers against every other dataset and calculate each guesser's mean success rate. Table \ref{tab:webhostSuccess} shows the mean success rate of each guesser: PCFGv4 and NN outperform others, demonstrating a relatively high degree of generalizability compared to others. The Identity guesser and OMEN perform notably worse. This is expected for the Identity guesser with its inability to generalize, but surprising for OMEN. There is a notable amount of variance in the success rates of guessers with similar approaches: 15\% difference between Markov models JtR and OMEN, and 10\% difference between PCFG-based guessers PCFGv4 and Sem. This highlights how even guessers with similar underlying approaches can display differing generalization behavior. 

\begin{table}[tb]
        \begin{center}
        \caption{Guessers' generalizability, with 300M guess cutoff. A higher success rate indicates a better ability to generalize.}
        \label{tab:webhostSuccess}
        \begin{tabular}[t]{lccccc}
        \toprule
        \textbf{Identity} & \textbf{OMEN} & \textbf{JtR-Markov} &  \textbf{Sem} & \textbf{NN} & \textbf{PCFGv4} \\
        \midrule
        15.378 & 15.664 & 30.265 & 33.099 & 39.585 & 43.618 \\
        \bottomrule
        \end{tabular}
         \end{center}
\end{table}


\begin{figure*}[h]
\begin{subfigure}{\columnwidth}
\includegraphics[width=0.77\columnwidth]{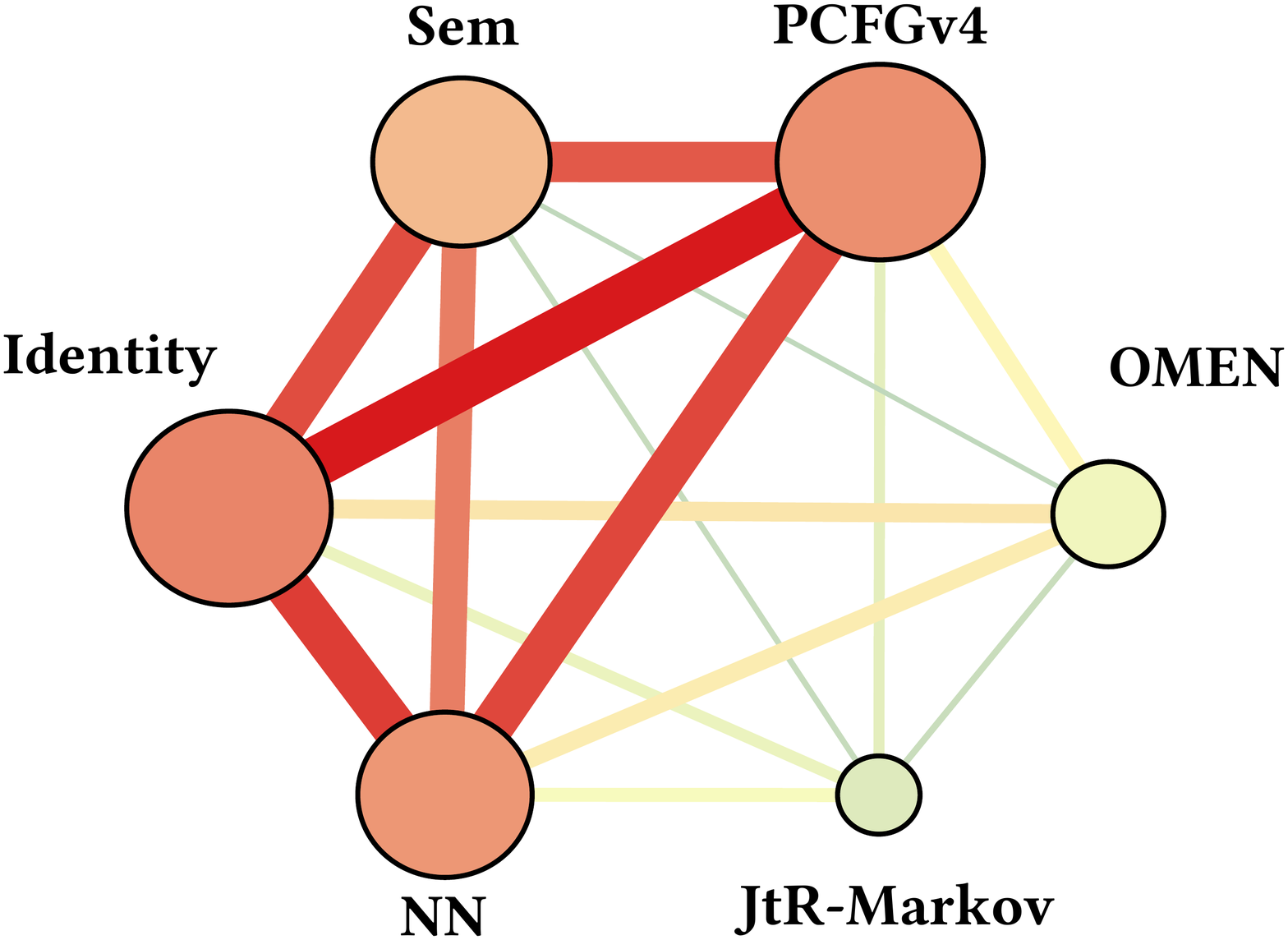}
\includegraphics[width=0.106\columnwidth]{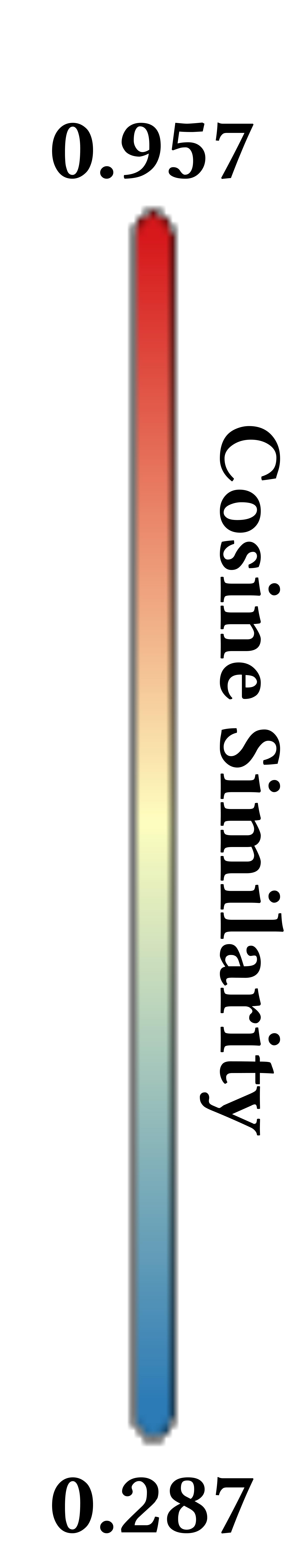}
\caption{Cosine guessing similarity with 1 million guesses.}
\label{fig:cosineGuessing1M}
\end{subfigure}
\hfill
\begin{subfigure}{\columnwidth}
\includegraphics[width=0.77\columnwidth]{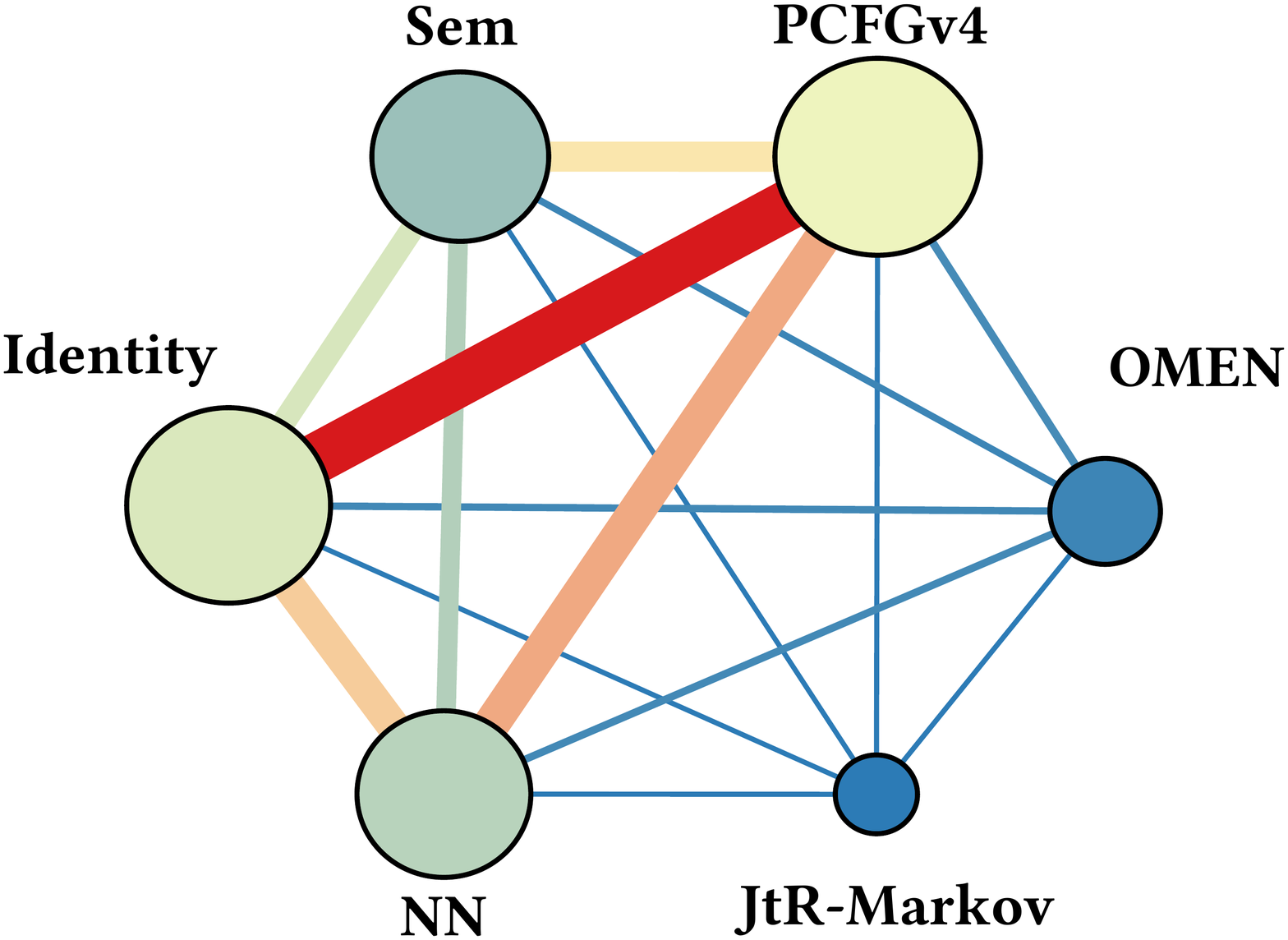}
\includegraphics[width=0.106\columnwidth]{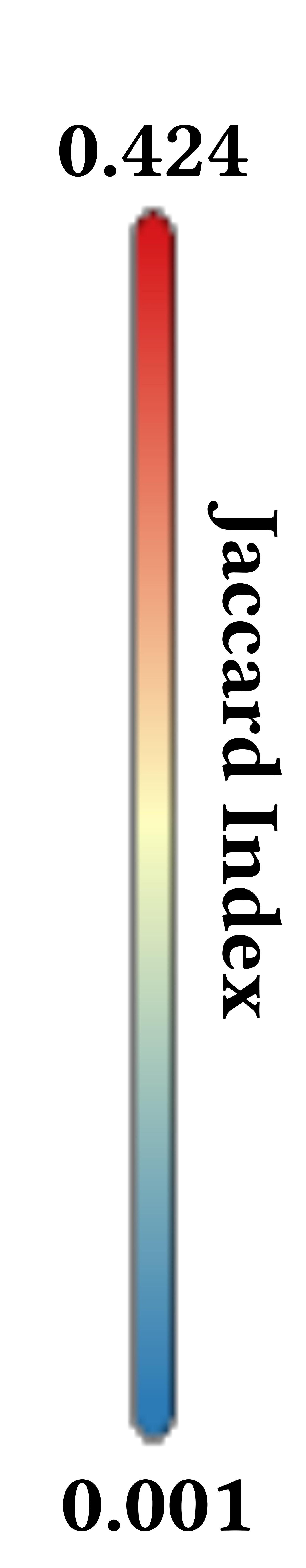}
\caption{Jaccard guessing similarity with 1 million guesses.}
\label{fig:jaccardGuessing1M}
\end{subfigure}
\begin{subfigure}{\columnwidth}
\includegraphics[width=0.77\columnwidth]{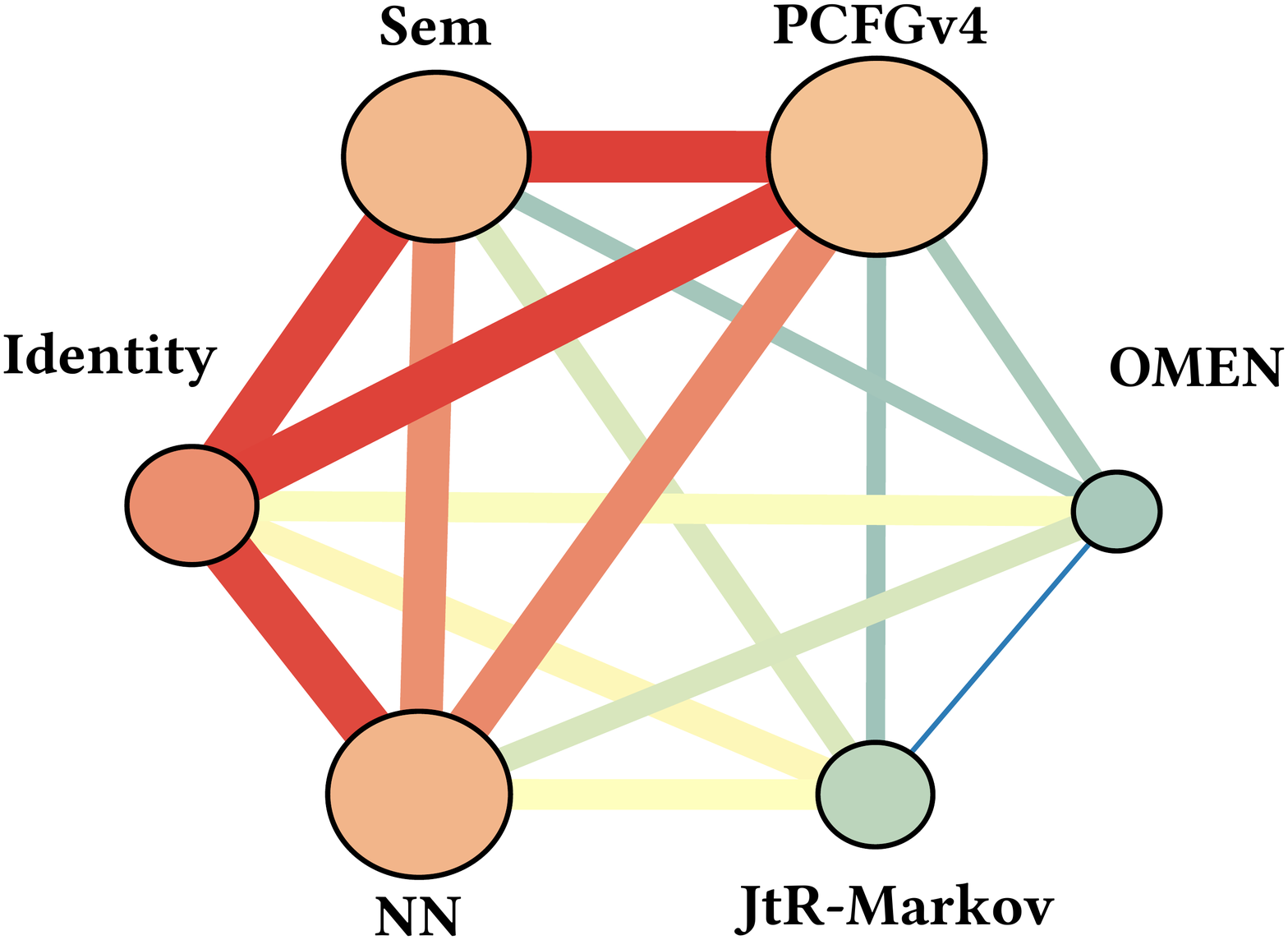}
\includegraphics[width=0.106\columnwidth]{images/ScaleBars/GuessingSimilarityCosine.pdf}
\caption{Cosine guessing similarity with 300 million guesses}
\label{fig:cosineGuessing300M}
\end{subfigure}
\hfill
\begin{subfigure}{\columnwidth}
\includegraphics[width=0.77\columnwidth]{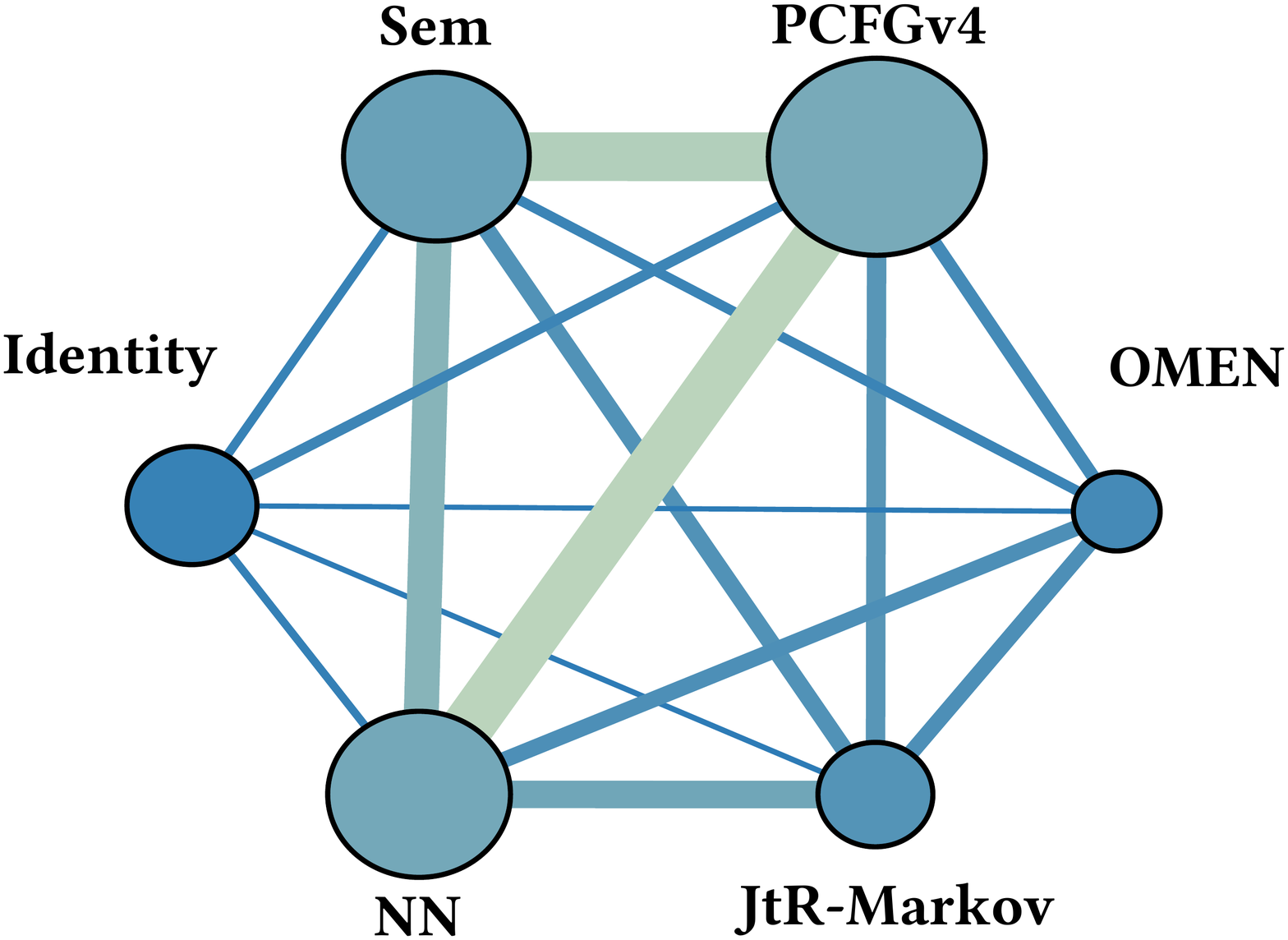}
\includegraphics[width=0.106\columnwidth]{images/ScaleBars/GuessingSimilarityJaccard.pdf}
\caption{Jaccard guessing similarity with 300 million guesses.}
\label{fig:jaccardGuessing300M}
\end{subfigure}
\caption{The cosine and Jaccard guessing similarity (see Eq.~\ref{eq:guessingSimilarity}) between guessers at the cutoffs of 1 million or 300 million guesses. The edge colors represent the similarity value between two guessers. The edge width further highlights the relative similarities  within a figure (thicker means more similar). The node size represents the guesser's average success rate. The node colors represent their average similarity.}
\label{fig:guessingSimilarity1-300}
\end{figure*}

\begin{table}[tb]
    \centering
    \caption{The mean guessing success rate (and standard deviation in parentheses) for each guesser when trained on different-sized subset of Twitter with a cutoff of 300M.}
    \label{table:guessingTrainingSize}
    \begin{tabular}{lccc}
    \toprule
    \multicolumn{1}{c}{} & \multicolumn{3}{c}{\textbf{Training Size}} \\ 
    \cmidrule[0.5pt](lr){2-4}
        \textbf{Guessers} & \textbf{1 Million} & \textbf{10 Million} & \textbf{30 Million}\\
        \hline
        Identity & 21.194(10.172)  & 33.441(13.458) & 39.853(14.186) \\
        JtR & 27.570(12.853) & 27.541(12.846) & 27.527(12.828) \\
        OMEN & 29.077(11.383) & 29.216(10.916) & 29.461(10.942) \\
        Sem & 41.493(13.669) & 46.910(14.432) & 48.021(14.832) \\
        PCFGv4 & 41.517(11.469) & 48.719(13.51) & 51.178(14.242) \\
        NN & 43.688(13.420) & 56.500(14.674) & 58.259(14.970) \\
        \bottomrule
    \end{tabular}
\end{table}
\subsubsection{Sensitivity to Training Size}
\label{subsec:trainingDataSize}
We intend to learn how each guesser's success rate is impacted by the size of training data, drawn from the same distribution. Sampling from the Twitter dataset\footnote{We train on Twitter for this purpose, as opposed to the Merged dataset, since the Merged dataset would contain the testing (target) data.}, we create three different datasets of sizes 1 million, 10 million and 30 million. After training guessers on each dataset, we generate guess lists at a cutoff of 300M and test them against all other datasets. Table \ref{table:guessingTrainingSize} reports the mean success rates by Eq.~\ref{eq:FixedDandGSuccess}. All guessers (except JtR-Markov) improve when trained on the larger dataset, but to various extents. The Identity guesser has the most drastic improvement with training size growth, from 21.2\% to 39.85\%. OMEN and JtR-Markov show the least improvement.
  Sem, PCFGv4, and NN have more modest, but notable improvements, increasing their success rates by 6.5\%, 9.7\%, and 14.6\%, respectively. 

\subsubsection{Guessing Similarity} 
Using our notion of guessing similarity (see Eq.~\ref{eq:guessingSimilarity}), we analyze how similar the guess lists of two guessers are when they are trained on the same training data. Figure \ref{fig:guessingSimilarity1-300} shows the cosine and Jaccard guessing similarity between guessers at cutoffs of 1 million and 300 million guesses. For both 
cutoffs, PCFGv4, Sem, ID and NN share high structural (cosine) similarity when compared to OMEN and JtR (see Figure \ref{fig:cosineGuessing1M} and Figure \ref{fig:cosineGuessing300M}). Interestingly, despite both deploying a Markov approach, JtR and OMEN are dissimilar. 
This is likely because OMEN outputs guesses in probability order, whereas JtR-Markov does not.

Figures \ref{fig:jaccardGuessing1M} and \ref{fig:jaccardGuessing300M} show Jaccard guessing similarity between guessers, capturing the overlap of guessers' guesses, at both cutoffs.  Guessers with higher success rates (see Table \ref{table:guesserSuccess}) seem to have higher Jaccard guessing similarity (or overlap): At 1 million, the two best guessers PCFGv4 and ID share the highest overlap whereas PCFGv4, Sem and NN with the highest success rates at 300 million have highest overlaps. One can also readily observe that the Jaccard guessing similarities decrease as the cutoff increases. This change suggests that by generating more passwords, each guesser has begun demonstrating their own unique guessing behavior (i.e., the percentage overlap between guessers' guess lists decreases).

\begin{figure}[tb]
\includegraphics[width=0.77\columnwidth]{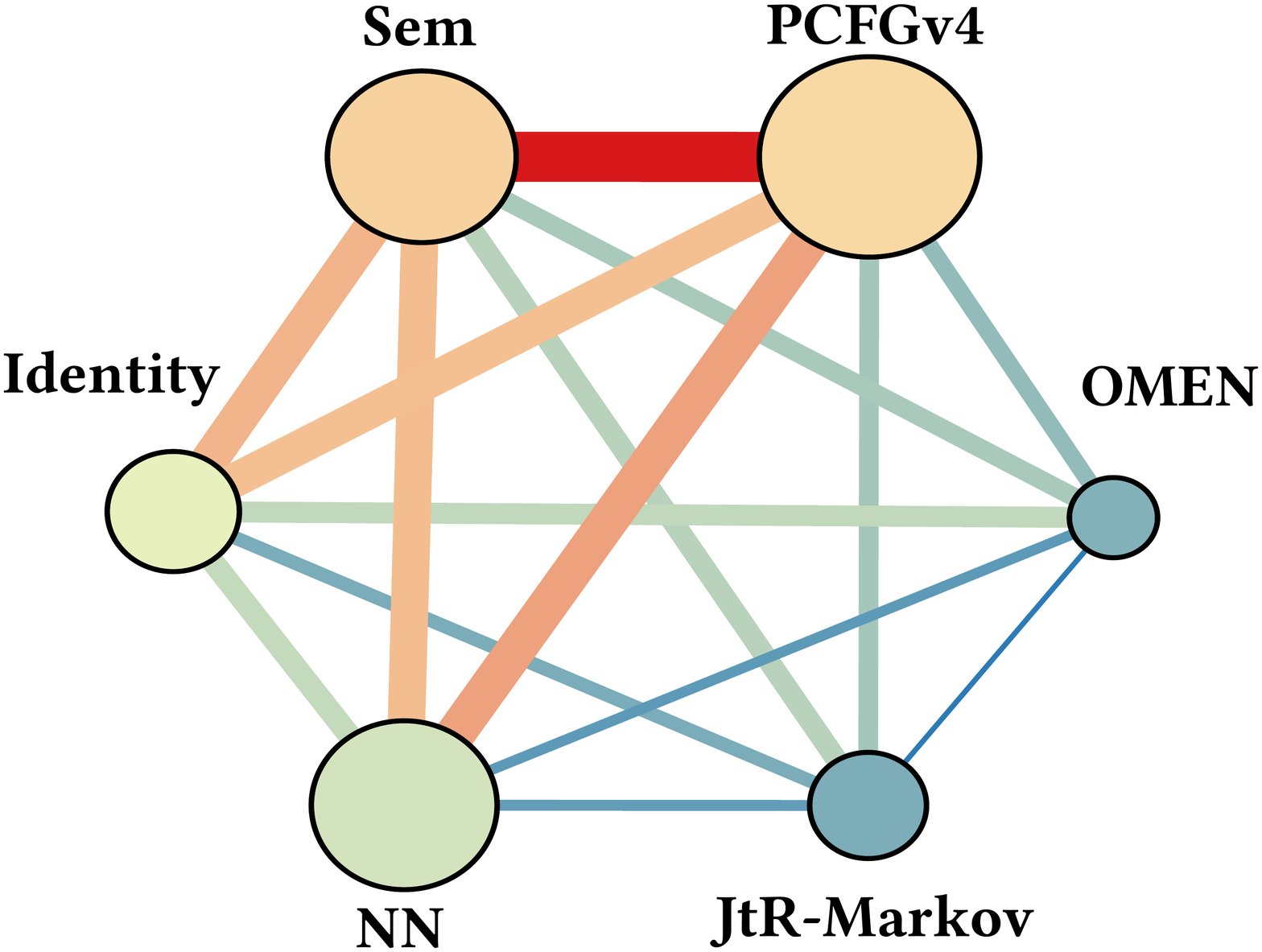}
\includegraphics[width=0.105\columnwidth]{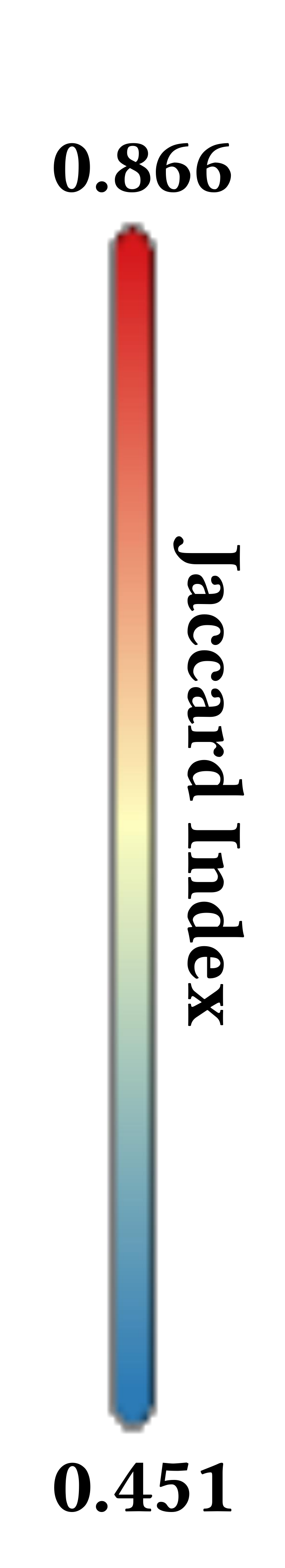}

\caption{The generalized Jaccard successful guessing similarity between guessers. The edge weights and colors represent the similarity of two guessers. The node size represents the guesser's average success rate. The node color represents the guesser's average similarity with other guessers.}
\label{guesserSuccessSimilarity}
\end{figure}
\subsubsection{Successful guessing similarity} Our guessing similarity analyses showed that guessers trained on the same data, generate mostly unique guesses (see Figures \ref{fig:jaccardGuessing1M} and \ref{fig:jaccardGuessing300M}). However, it is possible that many of these unique guesses are unsuccessful. In this light, one might be interested in measuring the uniqueness of \emph{successful} guesses between guessers. To achieve this, we use our successful guessing similarity metric in Eq.~\ref{eq:suc-guessing-similarity} with generalized Jaccard index.\footnote{The generalized Jaccard allows us to weight the successful guesses of each guesser based on their frequencies in the target dataset.} 

As shown in Figure \ref{guesserSuccessSimilarity}, there is still a considerable degree of uniqueness in successful guesses. Even Sem and PCFGv4---with the highest similarity---have a generalized Jaccard index of 0.86, implying that ~14\% of their successful guesses are unique to one guesser. Similarly, NN and Sem, by sharing ~72\% of their successful guesses, owe 28\% of their success to unique passwords. Interestingly, the Identity guesser seems to have moderate Jaccard similarity with any other guesser (i.e., its similarity values range from 0.529 to 0.725) despite its smaller guess lists sizes (i.e., ranging from 2.2 million to 40 million compared to 300 million for all other guessers).  These findings offer two important recommendations: (i) the use of one guesser does not make another guesser entirely redundant, even when the underlying approach or achieved success rates are similar; (ii) The cost-effective Identity guesser can complement any other guessers as it has a relatively high number of successful guesses. 
We explore the gains achieved by combining multiple guessers in our combination attack discussed below in Section \ref{subsec:combinationAttacks}.

\subsection{Combining Guessers}
We evaluate the ability of password guessers to complement one another on a previously unseen dataset (i.e., LinkedIn) in an offline attack scenario.  We begin in Section \ref{subsec:indivGuessers} by evaluating each individual guesser against the LinkedIn dataset. Next we analyze different combinations of guessers in Section \ref{subsec:combinationAttacks}.

\subsubsection{Individual Guessers} \label{subsec:indivGuessers}
To compare guessers' performance, we train each guesser on the Merged dataset, and allow them to each make 2 billion guesses against the LinkedIn dataset. 
As reported in Table \ref{table:LinkedIn}, NN outperforms all others, with a 4.3\% lead over PCFGv4.
PCFG-based (PCFGv4 and Sem) and Identity guessers outperform Markov-based guessers (OMEN and JTR-Markov). Figure \ref{fig:LinkedIn} depicts the percentage of guessed passwords over the number of guesses. JtR-Markov surpasses OMEN close to the end of the attack. Notably, PCFGv4, Identity, and NN traded places for the best guesser before Identity ran out of guesses. 
We next apply our findings from our successful guess similarity experiments to further improve the results using combination attacks.

%

\begin{table}[tb]
        \centering
        \caption{Percentage of LinkedIn passwords successfully guessed.  Guessers are trained on the Merged dataset and cutoff at 2 billion guesses.}
        \label{table:LinkedIn}
        \begin{tabular}{lccccc}
        \toprule
        \textbf{OMEN} & \textbf{JtR-Markov} & \textbf{Identity} & \textbf{Sem} & \textbf{PCFGv4} & \textbf{NN}\\
        \midrule
        35.641 & 37.028 & 47.561 & 55.159 & 58.798 & 63.145\\
        \bottomrule
        \end{tabular}
\end{table}

\begin{figure}[t]
    \begin{center}
        \includegraphics[width=\columnwidth]{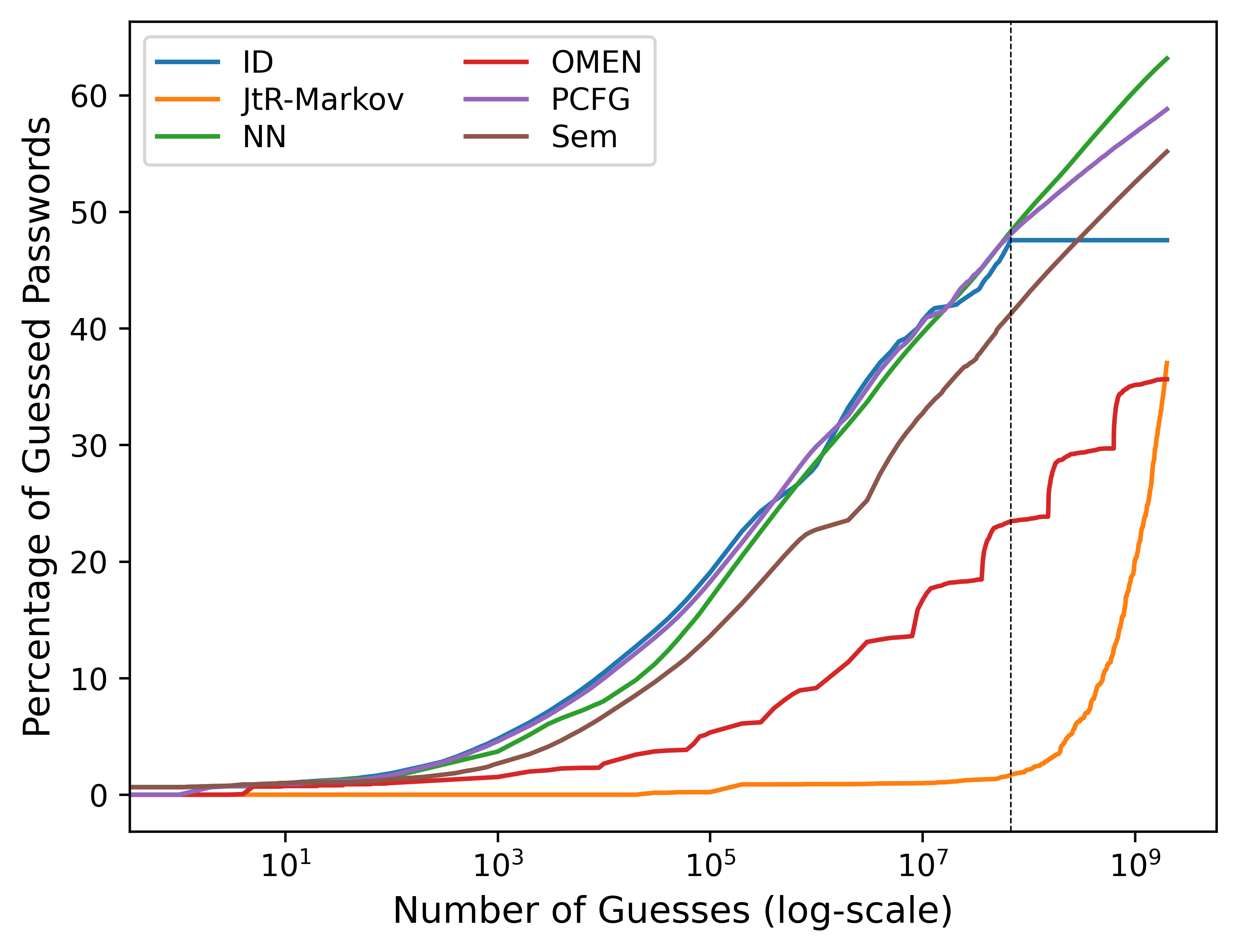}
    \vspace{-10pt}
    \caption{Performance of guessers trained on the Merged dataset and tested against LinkedIn. The dotted line marks the Identity guesser's last guess at 67 million guesses, each other guesser made 2 billion guesses.}
    \label{fig:LinkedIn}
    \vspace{-12pt}
    \end{center}
\end{figure}

\subsubsection{Combination Attacks}
\label{subsec:combinationAttacks}
Our analyses shed light on how guessers complement each other by generating unique successful guesses. We also learn that the Identity guesser not only complements every other guesser, but also often outperforms some advanced guessers. These findings motivate us to design a \textit{combination attack} where the Identity guesser is used to attack a password dataset prior to the application of a set of other guessers. This hybrid approach is recommended in John the Ripper where a traditional attack follows wordlist mode. We run many independent combination attacks on LinkedIn. 
Each guesser is trained on the Merged Dataset and produces two billion guesses. 

Table \ref{table:Combinations} reports the result of our combination attacks. When ID is combined with any individual guesser (e.g., ID+O, ID+J, etc.), the combination attacks experience a notable degree of improvement compared to an individual guesser's performance (compare the columns of sole guesser vs. ID + guesser). JtR-Markov experiences the largest improvement of 18.67\%. Even guessers with high success rates (e.g., NN and PCFGv4) realize improvements of 1\% to 4\%. By dramatically increasing the success rate of weaker guessers (e.g., OMEN and JTR-Markov), this combined approach makes less resource intensive guessers more competitive.

\begin{table*}[tb]
    \caption{The percentage of LinkedIn passwords cracked by an offline attack using the Identity guesser followed by a combination of guessers, each making two billion guesses. The names of guessers are shortened to their first letters: (P)CFG, (O)MEN, (N)N, (S)em, and (J)tR-Markov. Each combination attack is color-coded by its runtime for training and guess generation: Green is less than 8 hours (i.e., a workday), yellow is less than 16 hours, and red is over two weeks.}
    \label{table:Combinations}
    \setlength{\tabcolsep}{2pt}
    \begin{subfigure}{\linewidth}
        \centering
        \begin{tabular}{lccccccccccccc}
        \toprule
        \multicolumn{2}{c}{\textbf{Sole Guesser}}&
        \multicolumn{2}{c}{\textbf{ID + 1 Guesser}}&
        \multicolumn{4}{c}{\textbf{ID + 2 Guessers}}&
        \multicolumn{4}{c}{\textbf{ID + 3 Guessers}}&
        \multicolumn{2}{c}{\textbf{ID + 4 Guessers}}\\
        \cmidrule[0.7pt](lr){1-2}\cmidrule[0.7pt](lr){3-4}\cmidrule[0.7pt](lr){5-8}\cmidrule[0.7pt](lr){9-12}\cmidrule[0.7pt](lr){13-14}
        \textbf{guesser}&\textbf{guessed}&
        \textbf{guesser}&\textbf{guessed}&
        \textbf{guessers}&\textbf{guessed}&
        \textbf{guessers}&\textbf{guessed}&
        \textbf{guessers}&\textbf{guessed}&
        \textbf{guessers}&\textbf{guessed}&
        \textbf{guessers}&\textbf{guessed}\\
        \cellcolor{lgreen}OMEN & 35.641 & \cellcolor{lgreen}O & 52.272 & \cellcolor{lgreen}O+J & 57.628 & \cellcolor{lgreen}J+P & 65.038 & \cellcolor{lgreen}O+J+S & 63.825 & \cellcolor{pink!38}O+P+N & 67.336  & \cellcolor{yellow!25}O+J+S+P & 66.931 \\
        \cellcolor{lgreen}JtR-M & 37.028 & \cellcolor{lgreen}J& 55.693 & \cellcolor{lgreen}O+S& 61.536 & \cellcolor{pink!38}J+N& 65.909 &  \cellcolor{lgreen}O+J+P& 65.466  & \cellcolor{yellow!25}J+S+P& 66.614  & \cellcolor{pink!38}O+J+S+N& 67.484  \\
        \cellcolor{lgreen}Sem & 55.159 & \cellcolor{lgreen}S & 59.773 & \cellcolor{lgreen}O+P & 62.907  & \cellcolor{yellow!25}S+P & 63.866 & \cellcolor{pink!38}O+J+N & 66.199 & \cellcolor{pink!38}J+S+N & 67.247 & \cellcolor{pink!38}O+J+P+N & 68.060\\
        \cellcolor{lgreen}PCFGv4 & 58.798 & \cellcolor{lgreen}P & 61.158 & \cellcolor{pink!38}O+N & 65.241 & \cellcolor{pink!38}S+N & 66.411  & \cellcolor{yellow!25}O+S+P & 65.260  & \cellcolor{pink!38}J+P+N & 67.855 &\cellcolor{pink!38}O+S+P+N & 68.169 \\
        \cellcolor{pink!38}NN & 63.145 & \cellcolor{pink!38}N & 64.876 & \cellcolor{lgreen}J+S & 63.255 & \cellcolor{pink!38}P+N & 67.082 & \cellcolor{pink!38}O+S+N & 66.705 &\cellcolor{pink!38}S+P+N & 67.943 & \cellcolor{pink!38}J+S+P+N & 68.605 \\
        %
        \bottomrule
        \end{tabular}
        \end{subfigure}
\end{table*}

As shown in Table \ref{table:Combinations}, when more guessers are combined with the Identity guesser, the success rate increases, but with diminishing returns. 
For example, compare J to J+S (+7.562\%), J+S to J+S+P (+3.359\%), and J+S+P to J+S+P+N (+1.991\%). 
There seems to be two factors in determining which additional guesser can improve an existing combination attack the most: the success rate of the candidate guesser, and its successful guessing similarities with each of the combined guessers. A candidate guesser with higher success rate has more potential to improve the combined guesser (e.g., compare O+J to O+S). However, a candidate guesser with low successful guessing similarities can be a more effective addition. This interplay of success rate and successful guessing similarities might make a less successful guesser with lower successful guessing similarities more attractive. For example, the weaker JtR and stronger 
Sem have successful guessing similarities of 0.675 and 0.902 with PCFGv4. The addition of JtR to the combination attack of ID+P offers more improvement than the addition of Sem (3.88\% vs.  2.71\%).


Each additional guesser also incurs higher runtime and resource requirements. The attacks color-coded green in Table \ref{table:Combinations} could be completed within one workday (or 8 hours),  whereas the yellow and red color-coded attacks must be run overnight (within 8-16 hours) and over two weeks, respectively. The neural network is the largest contributor to runtime in our combinations and also adds GPU requirements. Interestingly, unlike the sole guesser attacks, the slower combination attacks don't always outperform the faster attacks. For example, the O+J+P attack (65.466\%) runs in under 8 hours while S+P (63.866\%) and O+S+P (65.250\%) take between 8-16 hours, and N (64.875\%) and O+N (65.241\%) take over 2 weeks. This result implies that competitive success rates can be achieved by  the combination of computationally-cheap guessers with less resources. These combination attacks serve as a competitive alternative for practitioners without access to GPU resources, or with time constraints to perform reactive checking (e.g., J+P attack outperforms N while running within a workday and without GPU resources).   

\section{Discussion and Recommendations} \label{sec:discussion}
We provide further context by presenting use cases of our framework and a set of recommendations based on our empirical results described in Section \ref{subsec:experimental-setup}. 

\subsection{Framework Use Cases}\label{sec:discussion:usecases}
Our framework is useful for both practitioners and researchers in:
\vskip 2mm
\noindent\textbf{(1) Evaluating new guessers and/or settings.} As new password guessers emerge, our framework can be applied to update our knowledge of how to best  combine password guessers, by adding the new guesser to $\mathcal{G}$ and recomputing the formulas in Sec.~\ref{sec:framework}. Our framework equips practitioners and researchers to assess whether or not emerging password guessers offer some complementary power to existing and deployed guessers. This supports a more informed selection of sets of password guessers for password checking. 
Using our framework supports evaluation beyond typical practices of simply benchmarking individual guesser's success rates. 
\vskip 2mm
\noindent\textbf{(2) Identifying effective training data for password checking.} For guessing scenarios where the target passwords are hashed and salted, our framework can still be applied with the generalized Jaccard index comparison metric, assuming that the salt of each hashed password is available to the administrator, as is typically the case. We describe how this can be accomplished in Proposition 1, the proof for which can be found in the Appendix.

\begin{proposition}\label{prop:hashed-jaccard}
Assuming a candidate training password list $A$ and salted \& hashed password list $B_h$, the generalized Jaccard index between $A$ and $B_h$ can be computed by:
\begin{equation}
    J(A,B_h)= \frac{\fmin(A,B_h)}
{|A| + |B_h|- \fmin(A,B_h)},
\label{eq:gen-jaccard-hashed}
\end{equation}
where
$$
\fmin(A,B_h) = \sum\limits_{\mathclap{w \in \supp{A}}} min\left(o(w,A), g(w, B_h)\right).
$$
Here, $g(w,B_h) = \sum_{y\in B_h}\ind{y = H(w+s_y)}$, $\ind{.}$ is the indicator function, $s_y$ and $H(.)$ are, respectively, the salt and hash function originally used for computation of the salted \& hashed password $y$. Here, $|A|$ and $|B_h|$ are the number of  passwords in $A$  and the number of salted \& hashed passwords in $B_h$. Also, $\supp{A}$ is the set of unique passwords in the training dataset $A$, and $o(w,A)$ is the number of occurrences of $w$ in $A$.
\end{proposition}

While the offline Identity attack success rate can be used as a proxy for measuring the similarity of a candidate training dataset $A$ with a salted \& hashed password list $B_h$, the generalized Jaccard index is more informative. For example, consider two candidate training datasets: $Y$ with hundreds of millions of entries, and dataset $Z$ with one million entries.  If they each achieve a 50\% success rate, dataset $Z$ should be considered more similar and selected as the best training set. However, the pure Identity attack success rate falls short in distinguishing $Z$ from $Y$ as opposed to the generalized Jaccard index which would assign a higher similarity score to $Z$.

\vskip 2mm
\noindent\textbf{(3) Identifying complementary guessers for longer offline attacks.} Guessers might offer different complementary patterns for short and long guessing sessions (e.g., online vs.\ offline attacks). Our experiments show the complimentary patterns of guessers for shorter sessions (up to 300 million guesses). To identify complementary guessers for longer attacks (e.g., approx.\ $10^{14}$ guesses),  our successful guessing similarity (recall Eq.\ \ref{eq:suc-guessing-similarity}) can be employed, in combination with Monte Carlo methods \cite{DellAmico2015}), to approximate complementary effects of guessers. Rather than directly computing $L_{ik} \cup D_{\ell}$ in Eq.\ref{eq:suc-guessing-similarity} by allowing guesser $g_i$ (trained on $D_k$) to generate the guess list $L_{ik}$, one can approximately determine the elements of $L_{ik} \cup D_{\ell}$ (i.e., the passwords that would be successfully guessed in a leaked plaintext testing dataset $D_{\ell}$) by: (1) setting the threshold $\tau$ for maximum number of guesses (2) for each password $w \in D_{\ell}$, estimate its \emph{guess number} (i.e., the minimum number of required guesses) using Dell'Amico et al.'s approach \cite{DellAmico2015} (3) if the guess number is less than the threshold $\tau$, it belongs to $L_{ik} \cup D_{\ell}$. We note that to apply this procedure, each guesser $g_i$ should be able to assign a probability to a password. 


\subsection{Recommendations}
Our work provides a number of practical recommendations (R1-R4) for practitioners auditing their passwords. Of course, this set of recommendations may update as more guessers and training datasets are analyzed using our framework. While our work can be directly applied to reactive checking, it has a natural extension to proactive checking, as guessers that generate probability scores for a given password can be applied as password meters. 

\vskip 1.15mm
\noindent \textit{\textbf{R1: Try publicly-available leaked passwords first.}} Our results show that an attacker can be relatively successful by applying the Identity guesser (i.e., the training data of leaked passwords as a guess list) before considering any advanced guessers.  This might seem a familiar concept, occasionally applied in practice (e.g., John the Ripper \cite{JohnTheRipper}). However, to the best of our knowledge, the impact and benefits of using an Identity guesser vs. other guessers has not been extensively quantified. For the first 1 million guesses (a number considered feasible for online attacks \cite{FHO2016_PushingString}), the Identity guesser along with PCFGv4 outperform more advanced guessers. For offline attacks, the Identity guesser performed surprisingly well; with only 22 million guesses, on average it achieved 64\% of the success rate of the top offline guesser PCFGv4 with 300 million guesses (see Table \ref{table:guesserSuccess}). Additionally, in our LinkedIn experiments, the Identity guesser, with 67 million guesses, had 75\% the success rate of the top guesser NN, with 2 billion guesses (see Table \ref{table:LinkedIn}). These experiments strongly suggest that the Identity guesser can achieve high guessing success rates, comparable to the top guessers, while using at least an order of magnitude fewer guesses.   
Thus, we strongly recommend that leaked password datasets should be the first priority in password checking. 

\vskip 1.15mm
\noindent \textit{\textbf{R2: Apply combinations of guessers.}} 
Our results for guessing similarity show that the majority of guesses produced by each guesser are unique, even when the underlying approach or success rate is similar.  Even for \emph{successful} guesses, each tested guesser is able to crack passwords that others overlook (e.g., the Identity guesser found millions of LinkedIn passwords overlooked by other guessers).  Our analysis indicates that no single guesser is able to completely substitute another, and they can complement each other when used together. However, some combinations are more effective than others. Our framework can be used to assist identification of complementary guesser combinations.  We also show how some combinations of guessers can have comparably high success rates with lower computing requirements. For example, in less than 8 hours, Identity + PCFGv4 + JtR-Markov can achieve a success rate that compares to Identity + NN (which takes about 2 weeks). Considering both success rate and computing requirements, our results from targeting LinkedIn passwords suggest that a reasonable strategy is to apply this ordering of guessers: Identity, PCFGv4, JtR-Markov, Sem, OMEN, NN. As discussed in Section \ref{sec:discussion:usecases}, our framework can be used to identify complementary combinations involving additional guessers, and also for long guessing sessions. 

\vskip 1.95mm
\noindent \textit{\textbf{R3: Train with datasets similar to target.}}
Our results show that when choosing training data, the similarity to the target data is an important factor.\footnote{These results confirm and complement previous findings \cite{Ji_password_2017} by employing different features, more and larger datasets, and more password guessers. We also show how similarity can be measured between a hashed \& salted target dataset and a plaintext candidate training set.} Thus, our dataset similarity metric can be used to decide on the most effective training dataset. The most effective dataset can be identified, even when the target dataset is hashed, as outlined in Section \ref{sec:discussion:usecases}. 
\vskip 1.95mm
\noindent \textit{\textbf{R4: Consider using less training data.}} 
Using more training data takes more computing resources and longer training times. Our results indicate that training dataset size does not correlate with guessing success rates.  Although when sampling from the same dataset (Twitter), we observed that data size can increase training effectiveness, the gains between 1 million and 30 million training passwords are not as large as one might expect. Therefore, if time or space constraints exist, a reasonable compromise would be to use a sample of training data from a dataset with high similarity (such as Twitter in our experiments).

\section{Conclusion and Future Work}


We provide an in-depth analysis of password guessers, revealing insights regarding 
when and how to use them (both alone and in combination). This work demonstrates that combinations of computationally-cheap guessers can be comparably effective to more resource-intensive guessers. Our work also points towards a set of 
recommendations for practitioners who use password checking tools.  

Our framework (i.e., various metrics and statistics) for comparing password guessers and training datasets can be utilized or extended by practitioners and researchers for future password studies. While our present work supports decisions about how to best combine password guessers, there remains some human interpretation of the results---i.e., our framework can help identify the guessers that are most dissimilar and have the highest success rate; however the final decision of which to combine should be made by the human involved (and consider computational efficiency as well).  As such, an interesting direction is to develop artificial intelligence algorithms to automate finding combinations of guessers with a maximum success rate under budgeted time and resource requirements. Our present work lays the foundations for such future directions.  Another interesting direction for future work is to explore how to summarize a large training dataset into a smaller dataset that trains guessers just as well.  Such a smaller training dataset would decrease training time and aim to maximize success rate.  

\bibliographystyle{IEEEtran}
\bibliography{cite.bib}
\clearpage
\appendix
\section{Appendix A: Proofs}\label{appx}
\begin{proof}[Proof of Proposition \ref{prop:hashed-jaccard}]
By Lemma \ref{lem:max-min-rel} and Lemma \ref{lem:min-simple}, the generalized Jaccard index between the password list $A$ and unhashed password list $B$ (which is not accessible) can be computed by:
\begin{equation}
 J(A,B)= \frac{\sum\limits_{\mathclap{w \in \supp{A}}} min\left(o(w,A), o(w, B)\right)}
{|A| + |B|- \sum\limits_{\mathclap{w \in \supp{A}}} min\left(o(w,A), o(w, B)\right)},
\label{eq:mid-equation-proof}
\end{equation}
Defining $g(w, B_h) = \sum_{y\in B_h}\ind{y = H(w+s_y)}$ for counting the number of occurrences of password $w$ in the salted \& hashed password list $B_h$, we note that $o(w,B)=g(w,B_h)$ and $|B|$ = $|B_h|$. So Eq.~\ref{eq:mid-equation-proof} is equivalent to:
$$
J(A,B_h)= \frac{\sum\limits_{\mathclap{w \in \supp{A}}} min\left(o(w,A), g(w, B_h)\right)}
{|A| + |B_h|- \sum\limits_{\mathclap{w \in \supp{A}}} min\left(o(w,A), g(w, B_h)\right)}.
$$
Letting $\fmin(A,B_h) = \sum\limits_{\mathclap{w \in \supp{A}}} min\left(o(w,A), g(w, B_h)\right)$, we derive Eq.~\ref{eq:gen-jaccard-hashed}.
\end{proof}

\begin{lemma}\label{lem:max-min-rel}
Let $o(w,A)$ and $o(w,b)$ be the number of occurrences of password $w$ in password lists $A$ and $B$ respectively.  We have 
\begin{align}
  &\sum\limits_{\mathclap{w \in \left(\supp{A}\cup\supp{B}\right)}} min\left(o(w,A), o(w, B)\right) =\nonumber\\
  &|A| + |B| - \sum\limits_{\mathclap{w \in \left(\supp{A}\cup\supp{B}\right)}} max\left(o(w,A), o(w, B)\right).
\label{eq:max-min-rel}
\end{align}
Here, $|A|= \sum_{w \in \supp{A}} o(w,A)$ and $|B| = \sum_{w \in \supp{B}} o(w,B)$ are the number of passwords in $A$ and $B$ respectively. Also, $\supp{A}$ is the set of unique passwords in $A$.
\end{lemma}

\begin{proof}
One can observe that for any two numbers $a$ and $b$:
\begin{align*}
  min\left(a, b\right)+max\left(a, b\right) =a + b.
\end{align*}
Using this equality, we can derive
\begin{align*}
 &\sum\limits_{\mathclap{w \in \left(\supp{A}\cup\supp{B}\right)}} \Big[min\left(o(w,A), o(w,B)\right)&\\
 & + max\left(o(w,A), o(w,B)\right)\Big]&\\
  &=\sum\limits_{\mathclap{w \in \left(\supp{A}\cup\supp{B}\right)}} o(w,A) + o(w, B)&\\
  &=\sum\limits_{\mathclap{w \in \left(\supp{A}\cup\supp{B}\right)}} o(w,A)+ \sum\limits_{w \in \left(\supp{A}\cup\supp{B}\right)} o(w, B)&\\
  &=\sum\limits_{\mathclap{w \in \supp{A}}} o(w,A)+ \sum\limits_{w \in \supp{B}} o(w, B).&\\
\end{align*}
The last equality holds as $o(w,A)=0$ when $w \notin A$ and $o(w,B)=0$ when $w \notin B$. By decomposing the first summation, we have shown
\begin{align*}
 &\sum\limits_{\mathclap{w \in \left(\supp{A}\cup\supp{B}\right)}} min\left(o(w,A), o(w, B)\right)&\\
 &+ \sum\limits_{\mathclap{w \in \left(\supp{A}\cup\supp{B}\right)}} max\left(o(w,A), o(w, B)\right)=|A|  + |B|,&
\end{align*}
where  $|A|= \sum_{w \in \supp{A}} o(w,A)$ and $|B| = \sum_{w \in \supp{B}} o(w,B)$. By rearranging the terms of this equality, we derive Eq.~\ref{eq:max-min-rel}.
\end{proof}

\begin{lemma}\label{lem:min-simple}
Letting $o(w,A)$ and $o(w,b)$ be the number of occurrences of password $w$ in password lists $A$ and $B$ respectively, 
\begin{align}
  &\sum\limits_{\mathclap{w \in \supp{A}\cup \supp{B}}} min\left(o(w,A), o(w, B)\right)&\nonumber\\
  &=\sum\limits_{\mathclap{w \in \supp{A}}} min\left(o(w,A), o(w, B)\right),&
\label{eq:simplified-min}
\end{align}
where $\supp{A}$ is the set of unique passwords in $A$.
\end{lemma}
\begin{proof}
Partitioning $\supp{A}\cup \supp{B}$ to two disjoint sets of $\supp{A}$ and $\supp{B}-\supp{A}$, we have
\begin{align*}
  &\sum\limits_{\mathclap{w \in \supp{A}\cup \supp{B}}} min\left(o(w,A), o(w, B)\right)=&\\ &\sum\limits_{\mathclap{w \in \supp{A}}} min\left(o(w,A), o(w, B)\right) +&\\ &\sum\limits_{\mathclap{w \in \supp{B}-\supp{A}}} min\left(o(w,A), o(w, B)\right)&.
\end{align*}
As $o(w,A)=0$ for $w \in \supp{B}-\supp{A}$, we have $min\left(o(w,A), o(w, B)\right)=0$ for all $w \in \supp{B}-\supp{A}$. So we have derived Eq.~\ref{eq:simplified-min}.

\end{proof}

\end{document}